\documentclass[12pt]{article}
\usepackage[utf8]{inputenc}
\usepackage{amsfonts, amsthm, amsmath, graphicx, enumerate, verbatim, amssymb, mathrsfs,hyperref,array,subfigure}
\usepackage{bm,bbold,dsfont}
\usepackage[authoryear]{natbib}
\usepackage{xcolor}

\newtheorem{thrm}{Theorem}
\newtheorem{lemma}{Lemma}
\newtheorem{prop*}{Proposition}

\newtheorem{prop}{Proposition}
\newtheorem{corollary}{Corollary}
\newtheorem{thm*}{Theorem}
\theoremstyle{definition}

\newtheorem{example}{Example}

\usepackage[top=1in, bottom=1in, left=1in, right=1in]{geometry}
\usepackage{booktabs}

\newcommand{\bA}{\mathbf{A}}
\newcommand{\bB}{\mathbf{B}}
\newcommand{\bC}{\mathbf{C}}
\newcommand{\bE}{\operatorname{\mathbf{E}}}
\newcommand{\be}{\mathbf{e}}

\newcommand{\bkappa}{\boldsymbol{\kappa}}
\newcommand{\BI}{\mathbb{I}}
\newcommand{\BT}{\mathbb{T}}
\newcommand{\bT}{\mathbf{T}}
\newcommand{\bV}{\mathbf{V}}

\newcommand{\bx}{\mathbf{x}}
\newcommand{\bxi}{\boldsymbol\xi}
\newcommand{\diag}[0]{\ensuremath\operatorname{diag}}

\newcommand{\indic}[1]{\onebb_{\left\{#1\right\}}}
\newcommand{\m}[1]{{\mathbf{#1}}}
\newcommand{\onebb}[0]{\m{\mathds{1}}}
\newcommand{\real}{\mathbb{R}}
\newcommand{\Rn}{\real^n}

\newcommand{\sF}{\mathscr{F}}
\newcommand{\sG}{\mathscr{G}}
\newcommand{\sR}{\mathscr{R}}

\newcommand{\trans}[0]{\ensuremath\intercal}
\newcommand{\vinf}{v^{\text{inf}}}
\newcommand{\vsup}{v^{\text{sup}}}
\newcommand{\Xomit}[1]{}

\pdfoutput=1
\title{Optimal Bailouts in Diversified Financial Networks}
\author{Krishna Dasaratha\thanks{Department of Economics, Boston University. email: {\tt krishnadasaratha@gmail.com}} \\ Santosh S. Venkatesh\thanks{ Department of Electrical and Systems Engineering, University of Pennsylvania. email: {\tt venkatesh@seas.upenn.edu}.} \\ Rakesh Vohra\thanks{Department of Economics \& Department of Electrical and Systems Engineering, University of Pennsylvania. email: {\tt rvohra@seas.upenn.edu}.}}

\date{\today}

\begin{document}
\maketitle
\begin{abstract}
Widespread default involves substantial deadweight costs which could be countered by injecting capital into failing firms. Injections have positive spillovers that can trigger a repayment cascade. But which firms should a regulator bailout so as to minimize the total injection of capital while ensuring solvency of all firms? While the problem is, in general, NP-hard, for a wide range of networks that arise from a stochastic block model, we show that the optimal bailout can be implemented by a simple policy that targets firms based on their characteristics and position in the network. Specific examples of the setting include core-periphery networks.
\end{abstract}

\section{Introduction}
The ties of debt and equity between firms enable efficient risk sharing  and capital allocation. They are also a conduit by which a negative shock to a small set of firms can be amplified to generate widespread default. Defaults involve substantial deadweight costs, including fire sales, early termination of contracts, administrative costs of government bailouts, and legal costs.

To mitigate these costs a regulator can inject capital into failing firms so as to prevent defaults.  Injecting capital into a firm has positive spillovers if, by paying back its obligations to others, that firm allows its counterparties  to meet their own obligations, triggering a repayment cascade. Given these spillovers, which firms should the regulator bailout so as to minimize the total injection of capital while ensuring solvency of all firms? Even in this simple form, absent moral hazard and asymmetric information, the problem is hard. The first difficulty arises from the presence of multiple equilibria. Who to target and how much to inject will depend upon the equilibrium outcome that prevails post-injection. Second, even if one fixes the choice of equilibrium (which we will), the incremental benefit of spillovers triggered by injecting capital into a firm depend on which other firms have received infusions. Indeed, most variations of the problem of determining who to bailout and by how much so as to achieve the desired equilibrium at minimum cost is NP-hard. For examples,  see \cite{jackpern}, \cite{klages2022optimal}, \cite{dong} and \cite{pk}. This eliminates the possibility of an optimal policy that can be described by a `simple' index rule. Such a rule would assign an index to a firm that depends only on its characteristics and easily computed `global' information such as how central they are in the network. The index would determine whether it should be bailed out. \cite{demange}, for instance, offers an index in this spirit to prioritize firms linked by debt. However, this index relies on the uniqueness of the underlying equilibrium and it is not invariant to injections of cash.


NP-hardness excludes a simple index policy for {\em all} networks but does not eliminate this possibility for special cases. We consider one such case where the  underlying financial linkages are via equity cross holdings as in  \cite*{elliott2014financial}.  Here, a firm's value depends on its cash endowment (in the form of
held primitive assets) and the shares it owns in other firms.\footnote{Such equity holding networks, are pervasive~\cite{shi2019internal}.} When a firm's value drops below a solvency threshold, it discontinuously imposes losses on its counter-parties, e.g., as distress costs.  The counter-parties in turn may drag other firms below the solvency threshold.

The network of cross holdings is modeled as random with arbitrary block structure called the stochastic block model (SBM). SBMs are a popular model in the statistical analysis of networks, see \cite{lee}. 
Firms are partitioned into finitely many blocks, and the probability that one firm holds shares in another depends only on the blocks containing each firm. Two important examples of the SBM are core-periphery networks and cross-border relations in a global financial system, with financial systems in different countries represented as different blocks in the network.\footnote{The following document such structures in a variety of financial networks: \cite{BechAta}, \cite{AKS}, \cite{BEST}, \cite{DL}, \cite{BBL}, \cite{PSV}, \cite{CFRS} and \cite{MR}.} We allow any sufficiently dense network structure between the blocks. The precise distributional assumptions of the SBM are discussed in Section \ref{sec:SBM} and can be interpreted as providing a model of diversified financial networks.

Given an amount of cash available to inject into the network, we wish to maximize the number of firms that achieve a value that exceeds the given solvency threshold. When the underlying network is a draw from an SBM, and the number of vertices is sufficiently large, we argue for a solution in terms of an index policy. Whether a firm is bailed out or not depends upon its cash endowment and how central the block to which it belongs is. The precise number of firms directly bailed out in each block is proportional to the block's Katz-Bonacich centrality, adjusted to account for the cost of rescuing firms in that block. Using this relationship we show that the number of firms that receive an injection increases with the cash available. However, the set of firms that receive an injection is {\em not} monotone in the sense of subset inclusion.
Furthermore, the relative dispersion in endowments, also affects the proportion of firms within a block that receive an injection relative to other blocks.

  We justify this policy by passage to a continuum version of the underlying network to bypass the difficulties associated with discrete network models such as ties and inessential corner cases.\footnote{This is not unusual in other contexts, see \citet{schmeidler} and \cite{azevedo2016supply}.} Continuum analogs of networks are
 called graphons, (see~\cite{lovasz2012large} and~\cite*{BCCH2017}) introduced in the previous decade as interest in large dense networks surged in a variety of different areas. They have been deployed to study 
targeting in networks, equilibria of network games, and the study of contagion (\cite{erol2020contagion}, and \cite{parise2023graphon}).

As our technique may be useful in other contexts, we outline the key steps involved.
\begin{enumerate}
    \item The first is to establish a concentration result. We show that with high probability, equilibrium firm valuations in any realization of an SBM  are concentrated around the equilibrium valuations of a particular deterministic network. Thus, the impact of bailouts on any {\em random} instance of  an SBM, must mirror that in this particular {\em deterministic} network.
    \item To determine the bailout policy in the deterministic network we pass to its graphon analog. The graphon analog admits a characterization of the extremal equilibrium firm valuations in terms of cutoffs in firms initial endowments. Within each block, there is a cutoff and firms in the block are solvent if and only if their endowments exceed that cutoff. This allows us to characterize the optimal bailout policy (assuming an extremal equilibrium) in terms of cutoffs for each block and relate each block specific cutoff to the centrality of the block. 
    \item The concluding step is to show that there is a bailout policy in the finite network that approaches the optimal policy in the graphon as the size of the network increases. 
\end{enumerate}

Section \ref{sec:SBM} describes the SBM of equity networks and the deterministic network that determines the equilibrium profile of valuations around which the equilibrium of each instance of an SBM concentrates.  It is this deterministic network that determines the graphon model analyzed in
Section  \ref{sec:model}. The characterization of extremal equilibrium valuations in terms of cutoffs will be found there. Section \ref{sec:infusion} characterizes the optimal cash injection policy in the graphon. We also show how this policy translates into a policy for a large finite network with approximately the same performance.

\section{The Stochastic Block Model}\label{sec:SBM}
In the equity cross-holding network of \citet{elliott2014financial} there are $n$ firms; each firm $i$ has a cash endowment $e_i$ and a (book) value $V_i$. A share $C_{ij} \in [0,1)$ of firm $j$ is held by firm $i$. Firms pay a bankruptcy cost $\beta$ if their values drop below a threshold $v^*$. The equilibrium values of firms satisfy the system of equations
\begin{equation}\label{dequ}
	V_i = \underbrace{e_i}_{\text{{\color{gray}endowment}}} 
            +  \underbrace{
            	\sum_{j=1}^n C_{ij} V_j
            }_{ \text{{\color{gray} cross-holdings}}}
        	- \underbrace{
        		\beta\onebb_{\{V_i < v^*\}}
        	}_{\text{{\color{gray}bankruptcy cost}}}
            \qquad (1\leq i\leq n)
\end{equation}
where $\onebb_{\{V_i < v^*\}}$ is the indicator function which takes value $1$ when $V_i < v^\ast$ and value $0$ when $V_i\geq v^\ast$. The system of equations is more compactly represented in vector notation: write $\bV = (V_1, \dots, V_n)^\trans$ for the (column) vector of valuations; $\be = (e_1, \dots, e_n)^\trans$ for the vector of endowments; $\onebb_{\{\bV < v^\ast\m{1}\}}$ for the vector indicator which picks out indices for which $V_i < v^\ast$; and $\bC = [C_{ij}]$ for the cross-holdings matrix interpreted as the adjacency matrix of an edge-weighted directed graph with self links permitted. The matrix of cross-holdings is not in general symmetric---there is no reason to believe that $C_{ij}$ equals $C_{ji}$ in general---and abeyant a cross-holding, we set $C_{ij} = 0$ if the corresponding edge $(i, j)$ is absent in the graph. In vector form the equilibrium firm valuations may be identified with the solutions of the non-linear fixed-point equation
\begin{equation}\tag{\ref{dequ}'}\label{dequ-vector}
	\bV = \be + \bC\bV - \beta\onebb_{\{\bV < v^\ast\m{1}\}}.
\end{equation}
There can be multiple equilibria and they form a lattice. We focus on the extremal (maximal or minimal) equilibria. 
        

\medskip\noindent\textsc{A stochastic block equity network}

\smallskip\noindent We will generate a random network via a stochastic block model and then construct equity cross-holdings based on this random network. The $n$ firms are partitioned into $m$ blocks representing distinct firm types; the blocks can differ in size. A random digraph (directed graph) is generated by inserting a directed edge between firms with a probability depending only on the blocks in which the firms reside: insert an edge directed from a firm in block $k$ to a firm in block $\ell$ with probability $g_{k\ell}$, the directed edges being independent of each other. Denote the realization of the multi-type random digraph by $\sG = \sG(n)$ and its adjacency matrix by $\bA = [A_{ij}]$ where $A_{ij}\thicksim\text{Bernoulli}(g_{k\ell})$ if $i$ is in the $k$th block and $j$ is in the $\ell$th block. The directed edges of $\sG$ represent directional equity sharing links.

We proceed to graft a regular equity-sharing model atop the engendered random graph $\sG$. Suppose $0<c<1$ is a fixed parameter of the system and, for each $j$, write $d_j := \sum_{i=1}^n A_{ij}$ for the in-degree of vertex $j$.\footnote{The possibility that $d_j = 0$ has vanishing probability asymptotically as long as any of the edge probabilities is non-zero and the blocks have sizes growing with $n$. We could finesse this nuisance possibility by conditioning on the high probability set $\{\,d_j > 0, 1\leq j\leq n\,\}$ or, simply, by just setting $A_{jj} = 1$ in the diagonal terms of the adjacency matrix $\bA$. Either approach will introduce a small amount of notational clutter that will have to be carried through the analysis without ultimately changing any of our reported results. We will accordingly unclutter presentation by ignoring this notational nuisance: suppose resolutely from now on that $d_j > 0$ for each $j$.}  We form a regular equity cross-holdings matrix $\bC = [C_{ij}]$ by setting
\begin{equation}\label{stochastic-block-model}
	C_{ij} = \frac{cA_{ij}}{d_j}.
\end{equation}
The interpretation of the model is that equity shares flow only across the directed edges in the graph $\sG$ with a fixed fraction $c$ of the equity of each firm $j$ shared evenly between the $d_j$ firms that are connected by an edge to $j$. The quantity $c$ is called {\em exposure} and  $1 - c$ represents that portion of each firm's valuation that is held by investors external to the network. This is also how random networks are modeled in \cite*{elliott2014financial}.

%

It is as well to clear up a terminological point here. The vector of valuations $\bV$, called ``book'' values, inflates actual firm valuations because of double counting. \cite*{elliott2014financial} sidestep this by focusing instead on the ``market'' values \begin{equation*}
	\widehat{\bV} = \diag\Bigl[1 - \sum\nolimits_i C_{i1}, \dots, 
		1 - \sum\nolimits_i C_{in}\Bigr] \bV,
\end{equation*}
which are the values held by outside investors. A firm's health is now assessed by its market value and in the~\cite*{elliott2014financial} formulation all that is requisite is that the indicator in the fixed-point equation~\eqref{dequ-vector} be replaced by $\indic{\widehat{\bV} < v^\ast_{\text{m}}\m{1}}$ where $v^\ast_{\text{m}}$ is a fixed market insolvency threshold. Specialized to our setting, market values are related to book values via the simple scaling relation $\widehat{\bV} := (1 - c) \bV$ and, as pointed out in~\cite*{amelkin2021}, the two formulations for the fixed-point equation are equivalent once $v^\ast$ is scaled appropriately: our book value insolvency threshold $v^\ast$ is related to the market value insolvency threshold $v^\ast_{\text{m}}$ via the simple scaling relation $v^\ast = v^\ast_{\text{m}}/(1 - c)$. If a firm defaults at some profile of book values in our setting, it will also do so with respect to its market value, and vice versa. Hence, for analytical purposes we may safely ignore the distinction between book and market values and focus on the book value equilibria of the fixed-point equation~\eqref{dequ-vector}.

\medskip\noindent\textsc{A block-regular clique}

\smallskip\noindent The random equity cross-holdings matrix $\bC$ thus constructed is generically, by design, asymmetric. While individual instantiations can vary wildly, the analysis of equilibria is simplified enormously by virtue of a concentration phenomenon in the theory of random matrices: \emph{the equilibrium valuations associated with almost all instances of $\bC$ concentrate at the equilibrium valuations associated with a fixed matrix of cross-holdings, $\overline{\bC}$, that we identify with a block-regular clique}. 

Some notation: write $\overline{d}_j=\sum_{i}\bE(A_{ij})$ for the expected in-degree of vertex $j$. Suppose block $k$ is comprised of $S_k$ vertices, $S_1 + \dots + S_m = n$. If $j$ is in block $\ell$, we may then identify
\begin{equation*}
	\overline{d}_j = g_{1\ell}S_1 + \dots + g_{m\ell}S_m =: \psi_{\ell}
\end{equation*}
and the expected in-degree of a given vertex is completely determined by the block in which it resides. Form the block-regular cross-holdings matrix $\overline\bC = \bigl[\overline{C}_{ij}\bigr]$ whose entries are given by
\begin{equation}\label{block-regular}
	\overline{C}_{ij} = \frac{c g_{k\ell}}{\psi_{\ell}}
\end{equation}
whenever $i$ is in block $k$ and $j$ is in block $\ell$. We call the corresponding equity sharing network a \emph{block-regular clique}; it is the extension of the \emph{regular clique} described in~\cite*{amelkin2021} to the stochastic block setting.

With the parametric setting unchanged---endowment vector $\be$, failure cost $\beta$, and threshold $v^\ast$---write $\overline{\bV} = \bigl[\overline{V}_j\bigr]$ for the equilibrium firm valuations associated with the block-regular clique. In analogy with~\eqref{dequ-vector}, $\overline{\bV}$ satisfies the fixed-point equation
\begin{equation}\label{dequ-vector-clique}
	\overline\bV = \be + \overline\bC\,\overline\bV
		- \beta\onebb_{\{\overline\bV < v^\ast\m{1}\}}.
\end{equation}
As before, there can be multiple solutions.



\medskip\noindent\textsc{Putative equilibria and concentration}

\smallskip\noindent The fixed-point equations~\eqref{dequ-vector} and~\eqref{dequ-vector-clique} induce a family of putative equilibria that we now describe. For each fixed binary vector $\bkappa\in\{0,1\}^n$ consider the fixed-point equations
\begin{equation}\label{dequ-putative}
	\begin{split}
    	\bV(\bkappa) &= \be + \bC\bV(\bkappa) - \beta(\mathbf{1} - \bkappa),\\
    	\overline{\bV}(\bkappa) &= \be + \overline{\bC}\,\overline{\bV}(\bkappa)
    		- \beta(\mathbf{1} - \bkappa).
    \end{split}
\end{equation}
We call each such $\bkappa$ a \emph{putative solvency vector} and think of it as identifying a putative list of solvent firms: firm $i$ is putatively solvent if $\kappa_i = 1$ and putatively insolvent if $\kappa_i = 0$. The solutions of the fixed-point equations~\eqref{dequ-putative} can be inconsistent in that the picture of solvency that emerges need not agree with $\bkappa$, hence the hedge ``putative''. A putative solvency vector $\bkappa$ is \emph{feasible} for the random cross-holdings matrix $\bC$  if, and only if, the corresponding fixed-point $\bV(\bkappa) = \bigl[V(\bkappa)_i\bigr]$ satisfies
\begin{equation*}
	V(\bkappa)_i 
		\begin{array}{ll}
			\geq v^\ast & \text{if $\kappa_i = 1$,}\\[2pt]
			< v^\ast & \text{if $\kappa_i = 0$.}
		\end{array}
\end{equation*}
Likewise, $\bkappa$ is feasible for the block-regular cross-holdings matrix $\overline\bC$ if, and only if, the corresponding fixed-point $\overline\bV(\bkappa) = \bigl[\overline{V}(\bkappa)_i\bigr]$ satisfies
\begin{equation*}
	\overline{V}(\bkappa)_i 
		\begin{array}{ll}
			\geq v^\ast & \text{if $\kappa_i = 1$,}\\[2pt]
			< v^\ast & \text{if $\kappa_i = 0$.}
		\end{array}
\end{equation*}
The equilibria associated with the fixed-point equations~\eqref{dequ-vector} and~\eqref{dequ-vector-clique} may hence be associated with the feasible solutions $\bkappa \in \{0,1\}^n$ of~\eqref{dequ-putative}.

A geometric viewpoint adds color to the picture. Each putative solvency vector $\bkappa$  induces a \emph{feasibility orthant} $\sR(\bkappa)$ of points $\bx = (x_1, \dots, x_n)$ in $\Rn$ satisfying
\begin{equation*}
	x_i \begin{array}{ll} \geq v^\ast & \text{if $\kappa_i = 1$,}\\[2pt]
			< v^\ast & \text{if $\kappa_i = 0$.} \end{array}
\end{equation*}
The feasibility orthants $\sR(\bkappa)$ partition $\Rn$ into $2^n$ regions as $\bkappa$ varies over $\{0,1\}^n$. A putative solvency vector $\bkappa$ is feasible for a given cross-holdings matrix if, and only if, the corresponding orthant $\sR(\bkappa)$ contains its putative solution: $\bkappa$ is feasible for $\bC$ (respectively, $\overline\bC$) if, and only if, $\bV(\bkappa)\in\sR(\bkappa)$ (respectively, $\overline\bV(\bkappa)\in\sR(\bkappa)$). 

The value of the putative fixed-point formulation~\eqref{dequ-putative} is that it replaces the \emph{non-linear} fixed-point formulation~\eqref{dequ-vector} which, in general, has a multiplicity of somewhat opaque solutions, by a system of $2^n$ \emph{linear} fixed-point equations each of which has a unique, completely characterized solution. The price to be paid for the analytical felicity conferred by this finesse is that not all the putative solutions are feasible: one has to winnow through the list to extract the equilibria.

We now consider large random networks and show a concentration result for putative equilibria: the values $\bV(\bkappa)$ are close to the values $\overline{\bV}(\bkappa)$. Reintroduce the explicit dependence on $n$ in the notation to keep the role of dimensionality firmly in view. Suppose $\be = \be(n)$ is any component-wise uniformly bounded sequence of endowment vectors and $\bkappa = \bkappa(n)$ is any sequence of putative solvency vectors. The insolvency cost $\beta$ and the number of blocks $m$ are held fixed, as are the inter-block edge probabilities $g_{k\ell}$, and we assume each block contains a non-vanishing fraction of firms. Consider now a sequence of multi-type random directed digraphs $\sG = \sG(n)$, the induced random cross-holdings matrices $\bC = \bC(n)$, and the corresponding block-regular cross-holdings matrices $\overline\bC = \overline\bC(n)$. Let $\bV(\bkappa) = \bV\bigl(\bkappa(n)\bigr)$ and $\overline\bV(\bkappa) = \overline\bV\bigl(\bkappa(n)\bigr)$ be the corresponding sequences of putative solutions of the fixed-point equations~\eqref{dequ-putative}.

Now introduce asymptotics. Our first result says that, for large $n$, the putative solution $\bV(\bkappa)$ in the stochastic block model is close to the putative solution $\overline{\bV}(\bkappa)$ in the block-regular clique. The following theorem says this somewhat more succinctly and accurately---and a lot more besides. The proof may be found in the appendix.
\begin{thrm}\label{thm:fixedkappa}
	$\|\bV\bigl(\bkappa(n)\bigr) - \overline{\bV}\bigl(\bkappa(n)\bigr) \|_\infty \to 0$
	almost surely as $n \rightarrow \infty$.
\end{thrm} 

With high probability we may conclude that, for any $\epsilon > 0$ chosen sufficiently small, $\|\bV(\bkappa) - \overline{\bV}(\bkappa)\|_{\infty} < \epsilon$, eventually, for all sufficiently large $n$. The result characterizes putative equilibria, but will be applied in the next section to characterize actual equilibria.

The proof shows, by applying Bernstein's inequality, that with high probability the spectrum of the matrix $\boldsymbol{C}$ of cross-holdings is close to the spectrum of the deterministic matrix $\boldsymbol{\overline{C}}$. The remainder of the analysis establishes that potential small perturbations to the matrix of cross-holdings do not have large impacts on the values of firms at the putative equilibrium $\bkappa$. A challenge relative to prior concentration results is that the network of connections is directed, so results on Hermitian matrices cannot be applied.

In the next step we finesse the usual computational complications associated with edge effects in discrete networks by passing to a continuum analogue of $\overline{\bC}$.

\section{The Graphon Model}\label{sec:model}
This section describes the graphon model as well as characterizes extremal equilibrium firm valuations in terms of cutoffs.

A \emph{kernel} is a bounded, symmetric, measurable function $C\colon (0,1]^2\to\mathbb{R}$.\footnote{It is more usual to define graphons on the Cartesian product of the closed unit interval $[0, 1]$ with itself. Our decision to work with the half-closed unit interval $(0, 1]$ instead is for notational convenience only. None of the results is altered by the addition or deletion of points at the boundary (or indeed, the addition or deletion of sets of Lebesgue measure zero) but the felicitous choice of the half-closed interval $(0, 1]$ as generator avoids the nuisance of having to introduce additional notation merely to handle an inconsequential measure zero case at the boundary.} Kernels  generalize weighted graphs in the sense that to each weighted graph  we may associate a distinct kernel $C$: if the graph has vertex set $[n]$, vertex weights $\alpha_i$, and edge weights $\beta_{ij}$, partition the unit interval into $n$ intervals $\BI_1, \dots, \BI_n$ where $\BI_i$ has length $\alpha_i\big/\sum_j\alpha_j$, and set $C(x,y) = \beta_{ij}$ if $(x, y)\in\BI_i\times\BI_j$. 

A kernel $C$ taking values in the unit interval, $0\leq C(x,y)\leq 1$, is called a \emph{graphon}, the name being a contraction of graph function. In the construction above, if the edge weights take values in the unit interval, then $C$ is a graphon. In the case of a simple, unweighted graph, the associated graphon $C(x,y)$ takes values in $\{0,1\}$. Viewed more broadly, graphons provide a natural generalization of Erd\"os--R\'enyi random graphs and, more generally, multi-type random graphs, via the intuitive interpretation of the values $C(x,y)$ as probabilities of links.


Conversely, a graphon can be interpreted as a limit of a sequence of graphs with an increasing number of vertices (see \cite{lovasz2012large}). 

\medskip\noindent\textsc{A graphon equity network}

\smallskip\noindent If we relax the symmetry requirement on $C(x,y)$ we obtain continuous analogs of directed graphs. This is the segue to the continuum germane for our purposes. Suppose, henceforth, that $C$ is a directed graphon, that is to say, $C\colon(0,1]^2\to[0,1]$ is bounded and integrable, not necessarily symmetric. To obviate trivialities, we suppose that, for each $y$, $C(x, y)$ is non-zero in some $x$-interval of positive measure.


Passing to a continuous limit of equity networks we may interpret the unit interval as a continuum of firms, each associated with a label $x\in(0,1]$, with the natural interpretation now that the (directed) graphon $C(x,y)$ denotes the share that firm $x$ holds in firm $y$. Then $I_y = \int_0^1 C(x, y)\,dx$ represents that fraction of the value of firm $y$ that is held by other firms in the network including firm $y$. Reusing notation from the discrete setting, we suppose $\sup_y I_y \leq c < 1$. This is the continuous analog of the sub-stochasticity condition on the cross-holdings. We may interpret $c$ as the (maximal) fraction of value held in-network with at least $1 - c$ of the value held by investors external to the network. Call such a $C$ an \emph{equity} graphon.

In the continuous limit, firm valuations $V_i$ are replaced naturally by valuation densities $v_x$ (with units of value/length): the total valuation of firms in an interval $\BI$ is given by $\int_{\BI} v_x\,dx$. Reusing notation, we are now led to write a formal continuous analogue of the fixed-point equilibrium equation~\eqref{dequ} in the form
\begin{equation}\label{eq:lineargraphon}
	v_x = e_x + \int_0^1 C(x, y) v_y\,dy - \beta\onebb_{v_x<v^*}.
\end{equation}
In this equation $v_x$ represents the valuation density at the point $x$ (with the lower case symbol serving to distinguish the continuous from the discrete), $e = e_x$ is bounded with $e_x$ representing the endowment density at $x$, and $\beta$ is the bankruptcy cost density.

An equilibrium in the equity graphon $C$ is associated with a solution of the fixed-point equation~\eqref{eq:lineargraphon}. While individual equilibria can have a rather wild and chaotic character, a standard application of the Knaster--Tarski theorem provides an elegant superstructure to the family of equilibria. Imbue any space of functions on the unit interval with the natural, pointwise partial order $\leq$: $u\leq v$ if $u_x\leq v_x$ for each $x$.
\begin{prop}\label{prop:Tarski}
	The set of equilibria of the fixed-point equation~\eqref{eq:lineargraphon} forms a 
	non-empty, complete lattice with respect to the pointwise partial order $\leq$.
\end{prop}
We conclude \emph{a fortiori} that the fixed-point equation~\eqref{eq:lineargraphon} has at least one solution, and indeed that there is a unique maximal equilibrium $v = \vsup$ and a unique minimal equlibrium $v = \vinf$. The maximal equilibrium pointwise dominates all other equilibria, while the minimial equilibrium is pointwise dominated by all other equilibria. The maximal and minimal equilibrium may coincide in which case~\eqref{eq:lineargraphon} has a single solution. Say that an equilibrium $v = v_x$ of the block equity graphon is \emph{extremal} if it is maximal, $v = \vsup$, or minimal, $v = \vinf$.





\medskip\noindent\textsc{A block equity graphon}

\smallskip\noindent Partition the unit interval into $m$ sub-intervals, each sub-interval representing a contiguous block of firms of a given type. Denote by $s_k$ the length of the $k$th sub-interval or block: this represents the fraction of all firms that are of type $k$. With $t_0:= 0$, write $t_k := s_1 + \dots + s_k$ for the partial sums and identify the sub-interval corresponding to firms of type $k$ by $\BT_k := (t_{k-1}, t_k]$. A generic firm is indexed by $x\in(0,1]$; it is of type $k$ if $x\in\BT_k$.


We suppose that endowments are type-dependent and piecewise smooth. Without loss we may take it that firms of each type are ordered by increasing endowment. Accordingly we begin with a family of type-specific endowment functions $\{\,f_k, 1\leq k\leq m\,\}$ where, for each $k$, $f_k\colon\BT_k\to\real^+$ is continuously differentiable and increasing in the sub-interval $\BT_k$, and stitch these functions together to create an endowment density
\begin{equation}\label{endowment-density}
	e_x = \begin{cases} f_1(x) & \text{if $x\in\BT_1$,}\\
			\hdotsfor{2}\\
			f_m(x) & \text{if $x\in\BT_m$,} \end{cases}
\end{equation}
which is piecewise continuously differentiable and increasing in every sub-interval. Intuitively, our formulation corresponds to firm values drawn randomly from distributions that can depend on the types $k$. The function $f_k(x)$ is then the CDF of this density. A leading case is $f_k(x)$ linear, which corresponds to values drawn uniformly from intervals $[a_k,b_k]$.

In analogy with the construction of the regular-block clique in the discrete setting, equity sharing links between firms are described via a function $G\colon(0,1]^2\to\real^+$ whose value at any given point $(x, y)$ depends only on the types of $x$ and $y$: $G(x, y) = g_{k\ell}$ if $x$ is of type $k$ and $y$ is of type $\ell$. With $0 < c < 1$ representing, as before, the in-network shared fraction of equity, we identify the \emph{block equity graphon}
\begin{equation*}
	C(x, y) = \frac{c G(x,y)}{\int_0^1 G(\xi, y)\,d\xi}.
\end{equation*}
Some notation clarifies matters: if $y\in\BT_\ell$ then, reusing notation from the discrete equity setting, the integral
\begin{equation*}
	\int_0^1 G(\xi, y)\,d\xi = g_{1\ell}s_1 + \dots + g_{m\ell}s_m =: \psi_{\ell}
\end{equation*}
is determined solely by the type of firm $y$. If firm $x$ is of type $k$ and $y$ is of type $\ell$, then
\begin{equation}\label{block-graphon}
	C(x, y) = \frac{cg_{k\ell}}{\psi_{\ell}} =: T_{k\ell}
\end{equation}
and a comparison with~\eqref{block-regular} shows that we have constructed the block equity graphon analogue of the block-regular clique.

The justification of the block equity graphon formulation~\eqref{block-graphon} rests in the fact that it can be viewed as the natural continuum limit of the stochastic block model equity network~\eqref{stochastic-block-model} of the previous section; the utility of the formulation rests in the computational simplicities that a passage to the continuum brings. 

For the graphon to give a reasonable approximation, we require that firms are small and the cross-holdings are sufficiently diversified. In practice, an important concern is that individual firms can generate systemic risk. Our model could be extended to allow a finite number of such large firms as well as a continuum of infinitesimal firms represented via a graphon. We can then study equilibria by jointly analyzing a discrete problem and a continuous problem, and the latter will be amenable to the techniques we develop in the remainder of the paper.

\medskip\noindent\textsc{Cutoffs and cutoff equilibria}

\smallskip\noindent In view of~\eqref{endowment-density} and~\eqref{block-graphon}, the generic fixed-point equation~\eqref{eq:lineargraphon} for equilibrium valuations reduces in the case of the block equity graphon to the family of equations
\begin{equation}\label{eq:univalue}
	v_x = f_k(x) + \sum_{\ell=1}^m T_{k\ell} \int_{\BT_\ell} v_y\,dy 
		- \beta\onebb_{v_x < v^\ast} \qquad (x\in\BT_k)
\end{equation}
as the firm of type $k$ traverses $1$ through $m$. The block structure of the graphon is manifested in the middle term on the right which is piecewise constant in each interval $\BT_k$. 

The example below will illustrate that a simple characterization of {\em all} equilibria of the block equity graphon will be elusive. Within any block, the intervals of solvency and insolvency can fluctuate wildly.  We conjecture our concentration results have analogues for larger classes of equilibria (e.g., when all firms above a cutoff $x_k^*$ are solvent and all firms below the cutoff are insolvent), but our subsequent focus will be on the extremal equilibria. 
\begin{example}\label{eg:swap}
	\emph{The swap}. For any (Lebesgue-) measurable set $\BI$ on the line, write $\BI + t$ 
	for its translation consisting of the points $x + t$ as $x$ varies over $\BI$. We take 
	$t$ to be larger than the diameter of $\BI$ so that $\inf (\BI + t) > \sup \BI$. 
	Continuum exemplars for $\BI$ are choices of tiny intervals while discrete exemplars 
	are singleton sets, though the structure does not preclude wilder sets. With this for 
	preparation, suppose now that both $\BI$ and $\BI + t$ are subsets of some interval 
	$\BT_k$. Introduce the nonce notation
	\begin{equation*}
		\Delta = \Delta(\BI, t) := \sup_{x\in\BI}\bigl(f_k(x+t) - f_k(x)\bigr).
	\end{equation*}
	The right-hand side is positive in view of the monotonicity of $f_k$.
	
	\emph{Swapping solvency and insolvency: I}. Suppose $v = v_x$ is an equilibrium of 
	the block equity graphon which satisfies
	\begin{equation}\label{eq:margin}
		v_x \begin{cases} 
				< v^\ast - \Delta & \text{if $x\in\BI$,}\\
				\geq v^\ast + \Delta & \text{if $x\in\BI + t$.}
			\end{cases}
	\end{equation}
	In words, firms in $\BI + t$ are solvent with a \emph{margin} $\Delta$ while firms in
	$\BI$ are insolvent with the same margin. Construct a new function $v' = v'_x$ by 
	setting
	\begin{equation}\label{eq:swap}
		v'_x =	\begin{cases}
					v_{x+t} - \bigl(f_k(x+t) - f_k(x)\bigr) & \text{if $x\in\BI$,}\\
					v_{x-t} + \bigl(f_k(x) - f_k(x-t)\bigr) & \text{if $x\in\BI + t$,}\\
					v_x & \text{otherwise.}
				\end{cases}
	\end{equation}
	In rough terms, $v'$ is constructed by swapping endowment-adjusted values of $v_x$
	and $v_{x+t}$ as $x$ varies over $\BI$ while keeping all other values fixed. The
	symmetry of the swap keeps the cross-share holding contribution due to the second
	term on the right in~\eqref{eq:univalue} unchanged: as
	$v'_x + v'_{x+t} = v_x + v_{x+t}$ for each $x\in\BI$, by grouping terms inside the
	integral, we have
	\begin{align*}
		\sum_\ell T_{k\ell}\int_{\BT_\ell} v'_y\,dy
			&= \sum_{\ell\neq k} T_{k\ell}\int_{\BT_\ell} v'_y\,dy
				+ T_{kk}\biggl[\int_{\BT_k\setminus[\BI\cup(\BI+t)]} v'_y\,dy
					+ \int_{\BI} (v'_y + v'_{y+t})\,dy\biggr]\\
			&= \sum_{\ell\neq k} T_{k\ell}\int_{\BT_\ell} v_y\,dy
				+ T_{kk}\biggl[\int_{\BT_k\setminus[\BI\cup(\BI+t)]} v_y\,dy
					+ \int_{\BI} (v_y + v_{y+t})\,dy\biggr]\\
			&= \sum_\ell T_{k\ell}\int_{\BT_\ell} v_y\,dy.
	\end{align*}
	Moreover, as $0<f_k(x + t) - f_k(x)\leq\Delta$ for each $x\in\BI$, in view 
	of~\eqref{eq:margin}, we see that 
	\begin{equation*}
		\onebb_{v_x' < v^\ast} =
			\begin{cases}
				\onebb_{v_{x+t} < v^\ast + (f_k(x+t) - f_k(x))} = 0
					& \text{if $x\in\BI$,}\\
				\onebb_{v_{x-t} < v^\ast - (f_k(x) - f_k(x-t))} = 1
					& \text{if $x\in\BI + t$,}\\
				\onebb_{v_x < v^\ast} & \text{otherwise.}
			\end{cases}
	\end{equation*}
	An easy algebraic verification now shows that $v'$ is another equilibrium of the block 
	equity graphon. Indeed, leveraging the swap~\eqref{eq:swap}, if $x\in\BI$, then firm
	$x+t$ is solvent (with margin $\Delta$), and so
	\begin{equation*}
		v'_x = f_k(x) + \sum_\ell T_{k\ell}\int_{\BT_\ell} v_y\,dy
			= f_k(x) + \sum_\ell T_{k\ell}\int_{\BT_\ell} v_y'\,dy
				- \beta\onebb_{v'_x < v^\ast}.
	\end{equation*}
	Similarly, if $x\in\BI + t$, then firm $x-t$ is insolvent (with margin $\Delta$), and so
	\begin{equation*}
		v'_x = f_k(x) + \sum_\ell T_{k\ell}\int_{\BT_\ell} v_y\,dy - \beta
			= f_k(x) + \sum_\ell T_{k\ell}\int_{\BT_\ell} v_y'\,dy 
				- \beta\onebb_{v'_x < v^\ast}.	
	\end{equation*}
	The remaining cases are trite as $v'_x = v_x$ if $x$ is not in $\BI$ or in $\BI + t$.
	Thus, the endowment-adjusted swap~\eqref{eq:swap} creates another equilibrium of the 
	block equity graphon, this time with the roles of insolvency and solvency interchanged 
	for $\BI$ and $\BI + t$.
	
	\emph{Swapping solvency and insolvency: II}. Reusing notation, suppose $v = v_x$ is
	an equilibrium of the block equity graphon where firms in $\BI$ are solvent and firms 
	in $\BI + t$ are insolvent:
	\begin{equation*}
		v_x \begin{cases} 
				\geq v^\ast & \text{if $x\in\BI$,}\\
				< v^\ast & \text{if $x\in\BI + t$.}
			\end{cases}
	\end{equation*}
	We don't need a margin anymore as with the roles of $\BI$ and $\BI + t$ reversed the
	endowments are working in our favor. Following the same argument as before, it is
	easy to verify that the swap~\eqref{eq:swap} yields a new equilibrium $v'$ which is 
	insolvent over $\BI$ and solvent over $\BI+t$. Or simply reverse the previous 
	construction.\hfill\qed
\end{example}

 While an arbitrary equilibrium can have complicated structures, the extremal points of the lattice of equilibria inherit a more regular character from the block structure of the graphon. These extremal equilibria are the most important members of the lattice of equilibria: they delineate bookend bounds on the overarching valuation picture and stability, and can also be selected via dynamic processes in which firms fail sequentially. We focus on extremal equilibria from this point onwards.
\begin{prop}\label{prop:block-graphon}
	If $v$ is an extremal equilibrium of the block equity graphon, then, for each 
	$k$, the restriction $v\big|_{\BT_k}$ of $v$ to $\BT_k$ is \emph{(i)} increasing and
	\emph{(ii)} continuous from the right with at most a single point of jump. Moreover, \emph{(iii)} if $v\big|_{\BT_k}$ has a jump at the point $x_k^\ast$ in $\BT_k$, then $v_{x_k^\ast} - v_{x_k^\ast-} = \beta$ and, furthermore, $v_{x_k^\ast} = v^\ast$ if $v$ is maximal, and $v_{x_k^\ast} = v^\ast + \beta$ if $v$ is minimal.
\end{prop}
We proffer two observations and two definitions motivated by them and defer the elementary proof to the appendix.
\begin{enumerate}
	\item If both $\vsup$ and $\vinf$ have jumps at points, say, $\overline{x}_k^\ast$
	and $\underline{x}_k^\ast$, respectively, in an interval $\BT_k$, then
	$\overline{x}_k^\ast < \underline{x}_k^\ast$. It follows that if either extremal equilibrium has a point of jump in any interval $\BT_k$ then the extremal equilibria are distinct. As a corollary, if the maximal and minimal equilibira coincide then the firms of each given type are either all solvent or all insolvent.
		
	\item Part (iii) of the theorem is where the graphon model proves useful. While the extremal equilibria are monotone in each block in the discrete setting as well, it is not generically true in the discrete setting that the size of the jump achieves exactly the distress cost $\beta$. Nor is it true in the discrete setting that at the point of jump the maximal equilibrium achieves a valuation exactly equal to the insolvency threshold $v^\ast$ or that the minimal equilibrium achieves a valuation exactly equal to $v^\ast+\beta$. The elimination of such nuisance edge effects in the continuum simplifies analysis and clarifies conclusions.
	
	\item In view of (iii), if $v\big|_{\BT_k}$ has a jump point at $x_k^\ast$, then 
	$x_k^\ast = \inf\{\,x\in\BT_k: v_x\geq v^\ast\,\}$ and we identify $x_k^\ast$ with the 
	\emph{smallest} index at which a firm is solvent. A jump point $x_k^\ast$, if one 
	exists in the interval $\BT_k$, of an extremal equilibrium represents a cutoff 
	 where the character of the extremal equilibrium valuation $v$ changes abruptly. 
	View these cutoffs more expansively to include intervals in which $v$ does 
	not have a jump by setting $x_k^\ast := t_{k-1} = \inf\BT_k$ if all firms in $\BT_k$ 
	are solvent and $x_k^\ast := t_k = \sup\BT_k$ if all firms in $\BT_k$ are insolvent. 
	We call  $\{\,x_k^\ast: 1\leq k\leq m\,\}$ the set of \emph{cutoffs} 
	(associated with the extremal equilibrium $v$).
	\item If a cutoff $x_k^\ast$ is at the boundary then the firms in $\BT_k$ 
	are either all solvent or all insolvent.\footnote{If one were to be punctilious, in the 
	latter case one must account for the possibility of a jump exactly at 
	$t_k = \sup\BT_k$, which occurs if $v(t_k-) = v^\ast$ and $v(t_k) = v^\ast + \beta$, 
	by modifying the language to exclude the measure zero point $\{t_k\}$ but this seems 
	obsessively pedantic. It does not change the general tenor of the observation.} The 
	case where $x_k^\ast$ is interior in $\BT_k$ is more interesting. Informally speaking, 
	in this case there is a cutoff below which all firms of type $k$ are insolvent and 
	above which all firms of type $k$ are solvent: 
	\begin{equation*}
		v_x \begin{cases} 
				< v^\ast & \text{if $t_{k-1}<x < x_k^\ast$,}\\[2pt]
				\geq v^\ast & \text{if $x_k^\ast\leq x\leq t_k$.}
			\end{cases}
	\end{equation*} 
	In a more graphic language, we say that an extremal equilibrium is a 
	\emph{cutoff equilibrium} if all its cutoff are interior.
\end{enumerate}

\medskip For a cutoff equilibrium to be computationally useful, the cutoffs should have a convenient characterization. As an illustrative example, we consider the simplest case when the endowment density linear within each block.

\begin{example}\label{eg:PiecewiseLinear}
	Suppose the endowment density is piecewise linear with each $f_k$ increasing linearly
	from $a_k$ to $b_k$ in the interval $\BT_k$:
	\begin{equation*}
		f_k(x) = a_k + (b_k - a_k)\cdot\frac{x - t_{k-1}}{t_k - t_{k-1}} 
			\qquad (t_{k-1} < x \leq t_k).
	\end{equation*}
    If $v = v_x$ has a jump in the interior of the interval $\BT_k$ (\emph{a fortiori}, if 
	it is a cutoff equilibrium) then the endowment $e_{x_k^*} = f_k(x_k^\ast)$ of the 
	cutoff firm is linear in the cutoff type $x_k^*$. In view of Proposition~\eqref{prop:block-graphon}, $v_{x_k^\ast} = v^\ast$, and so the fixed-point equation~\eqref{eq:univalue} yields 
 \begin{equation}
     v^\ast = a_k + (b_k - a_k)\cdot\frac{x_k^* - t_{k-1}}{t_k - t_{k-1}}   + \sum_{\ell=1}^m T_{k\ell}\int_{\BT_\ell} v_y\,dy.
 \end{equation}
 Suppose there is a jump discontinuity $x_k^*$ in each block; the subsequent equations are easily adapted to relax this assumption. Iteratively substituting for $v_y$ using equation \eqref{eq:univalue} (details can be found in the proof of Theorem \ref{t:spillovermatrix}), and using $\bT$ to denote the matrix $[T_{k\ell}]$, we obtain the linear equations 
 $$v^*= a_k + (b_k - a_k)\cdot\frac{x_k^* - t_{k-1}}{t_k - t_{k-1}}  + \left[\left(I-\bT D\right)^{-1}\bT \cdot \begin{pmatrix} t_1\cdot \frac{a_1+b_1}{2} - x^*_{1} \beta \\  (t_2-t_1)\cdot \frac{a_2+b_2}{2} - (x^*_{2}-t_1) \beta \\ \vdots \\  (t_m-t_{m-1})\cdot \frac{a_m+b_m}{2}- (x^*_{m}-t_{m-1}) \beta  \end{pmatrix}\right]_k.$$
 Here $D$ is the diagonal matrix whose $i^{th}$ entry is $t_i-t_{i-1}.$
 The upshot is that the linearity of $f_k(x)$ implies the cutoffs $x_k^*$ are characterized by a system of $m$ linear equations.\hfill \qed
\end{example}

\medskip\noindent\textsc{Connection to the SBM}

\smallskip\noindent The justification for the introduction of the block equity graphon is that it is, in a certain formal sense, the natural continuous limit of the SBM and replicates its equilibrium structure. The construction is standard though the concentration argument demonstrating asymptotic equivalence requires technical finesse.

As a \emph{gedanken} experiment, consider a sequence of SBM networks engendered by a sequence of multi-type random digraphs $\sG'(n)$. We suppose that the number of types $m$ is fixed, as are the inter-block link probabilities $g_{k\ell}$, and that the number of firms of each type grows linearly with $n$. In the notation of the previous section, let $S_k = S_k(n)$ be the number of firms of type $k$ in the graph $\sG'(n)$. Then there exist positive real values $s_1$, \dots, $s_m$ with $s_1 + \dots + s_m = 1$ such that $\tfrac1n S_1(n)\to s_1$, \dots, $\tfrac1n S_m(n)\to s_m$ as $n\to\infty$.

The block equity graphon~\eqref{block-graphon}, parametrized by $s_1$, \dots, $s_m$, and the link probabilities $g_{k\ell}$, is the natural limit object. Starting with the block equity graphon we can reverse the procedure and construct a sequence of random equity networks by sampling from the continuum model. For given $n$, begin by identifying vertices $i$ of a random graph $\sG(n)$ on $n$ vertices with equispaced points in the unit interval: $i\in\bigl\{0,\frac{1}{n-1},\frac{2}{n-1},\dots,1\bigr\}$. For each ordered pair of vertices $(i, j)$ insert a directed edge from $i$ to $j$ with probability $G(i,j)$, where as before, the link probability function $G(x, y)$ takes value $g_{k\ell}$ when $x\in\BT_k$ and $y\in\BT_\ell$. Directed edges are drawn independently. The resulting multi-type random graph $\sG(n)$ has approximately $(n - 1)s_k$ vertices of type $k$ and engenders a stochastic block equity network of $n$ firms via the regular equity-sharing formulation~\eqref{stochastic-block-model}. 

The number of vertices of type $k$ in the sampled multi-type directed random graph $\sG(n)$ differs from $S_k(n)$ only in the addition or deletion of a bounded number of points near the boundary of the interval $\BT_k$. Hence, the graphs $\sG(n)$ and $\sG'(n)$ are stochastically asymptotically equivalent. In this sense the block equity graphon~\eqref{block-graphon} is the continuous limit of the stochastic block equity sequence. But more can be said: the graphon mimics the asymptotic equilibrium structure as well.



At an extremal equilibrium in the stochastic block equity sequence engendered by the sampled sequence of multi-type random digraphs $\sG(n)$, let $\overline{i}_k(n)$ be the maximal index of a firm of type $k$ that fails and let $\underline{i}_k(n)$ be the minimal index of a firm of type $k$ that does not fail.

We will need a notion of \emph{stability} to ensure that perturbation arguments from a Picard-style iteration converge. In the interests of providing an early preview of the type of result that is achievable we state it sans specifics: we reserve the definition of stability to the next section where it arises more naturally in the consideration of a fictitious dynamic in the context of the impact of cash infusions, and defer the bulk of the proof to the appendix.

%
%


\begin{thrm}\label{thm:maximaleq}
	Suppose an extremal equilibrium of~\eqref{eq:univalue} $v = v_x$ is stable. Then along the sampled stochastic block equity sequence, $\overline{i}_k(n)$ and $\underline{i}_k(n)$ converge almost surely to the cutoffs $x_k^*$ for all $k$.
\end{thrm}

The theorem says that in large finite random networks, the set of solvent and insolvent firms at an extremal equilibrium are approximately determined by the cutoff types at the maximal equilibrium in the graphon from which they have been sampled.

We briefly describe the idea of the proof. The maximal equilibrium on the graphon is defined by a cutoff $x_k^*$ in each block, and we can define a putative equilibrium in a finite network by declaring firms solvent above and insolvent below these cutoffs. A key ingredient in the proof is Theorem~\ref{thm:fixedkappa}, which we use to show this putative equilibrium is likely to be ``close'' to being a feasible equilibrium on large finite networks. There are, however, small differences in values between this putative equilibrium and the graphon maximal equilibrium. The stability condition implies these differences do not lead to large spillovers, and we show there is indeed a feasible equilibrium near the putative one. This establishes a lower bound on the set of solvent firms at the maximal equilibrium on a large finite network. Similar arguments show that given a sequence of equilibria on growing finite networks, we can construct a nearby equilibrium on the graphon. This gives an upper bound on the set of solvent firms at the maximal equilibrium on a large finite network, and these lower and upper bounds together imply Theorem~\ref{thm:maximaleq}.

\section{Cash Infusions}\label{sec:infusion}

We use the graphon model to determine the optimal way to inject $K$ units of cash so as to maximize the fraction of firms that are solvent in a maximal equilibrium.  A \emph{cash infusion with budget $K$} is an integrable function $\iota(x)$ with $\int_0^1 \iota(x)\,dx= K$. Think of $\iota(x)\,dx$ as representing the amount of cash provided to firms in an infinitesimal interval $(x, x+dx)$ at the point $x$.

\medskip\noindent\textsc{The spillover matrix}

\smallskip\noindent The challenge is to quantify the spillovers generated by injecting capital into a particular firm. This is where the graphon model comes into its own. We use it to compute for each block $k$, the share of firms in block $k$ whose values drop below the solvency threshold if an \emph{infinitesimal} share of firms in block $k'$ fail. The ability to conduct marginal analyses of this kind are impossible in a discrete set up.

Mirroring our consideration of putative equilibria of stochastic block networks in~\eqref{dequ-putative}, for any given Boolean function $\kappa\colon[0,1]\to\{0,1\}$, consider the graphon fixed-point equation
\begin{equation*}
	v(\kappa)_x = e_x + \int_0^1 C(x,y) v(\kappa)_y\,dy - \beta\bigl(1 - \kappa(x)\bigr)
\end{equation*}
which, in the setting of the block equity graphon~\eqref{block-graphon}, simplifies to
\begin{equation}\label{putative-graphon-equilibrium}
	v(\kappa)_x = f_k(x) + \sum_{\ell} T_{k\ell} \int_{\BT_\ell} v(\kappa)_y\,dy
		- \beta\bigl(1 - \kappa(x)\bigr) \qquad (x\in\BT_k)
\end{equation}
for $1\leq k\leq m$. Borrowing from our language in the stochastic block setting, we call each such $\kappa$ a \emph{putative solvency function} and think of it as identifying a putative set of solvent firms in the continuum: firm $x$ is putatively solvent if $\kappa(x) = 1$ and putatively insolvent if $\kappa(x) = 0$. We naturally call any solution $v(\kappa) = v(\kappa)_x$ of~\eqref{putative-graphon-equilibrium} a \emph{putative equilibrium}. Following in the same vein as before, a putative solvency function $\kappa$ is \emph{feasible} if
\begin{equation*}
	v(\kappa)_x \begin{array}{ll} \geq v^\ast & \text{if $\kappa(x) = 1$,}\\
					< v^\ast & \text{if $\kappa(x) = 0$,} \end{array}
\end{equation*}
in which case $v(\kappa) = v(\kappa)_x$ is an equilibrium solution of~\eqref{eq:univalue}.

For obvious reasons we focus on putative solvency functions that correspond to cutoff equilibria. In a mild abuse of notation, for each vector $\bxi = (\xi_1, \dots, \xi_m)\in\BT_1\times\dots\times\BT_m$, write $\kappa(\,\cdot\, ; \bxi)$ for the indicator function parametrized by $\bxi$ which satisfies
\begin{equation*}
	\kappa(y; \bxi) = \begin{cases} 0 & \text{if $t_{k-1} < y < \xi_k$,}\\
						1 & \text{if $\xi_k\leq y\leq t_k$,} \end{cases}
\end{equation*}
as $k$ varies from $1$ through $m$. The associated putative equilibrium $v\bigl(\kappa(\,\cdot\,; \bxi)\bigr)$ solves the system of equations~\eqref{putative-graphon-equilibrium} with $\kappa = \kappa(\,\cdot\,; \bxi)$. If $v = v_x$ is a cutoff equilibrium valuation density with cutoffs $\bx^\ast = (x_1^\ast, \dots, x_m^\ast)$, then $\kappa(\,\cdot\,; \bx^\ast)$ is called a {\em feasible solvency function} and the associated putative equilibrium $v\bigl(\kappa(\,\cdot\,; \bx^\ast)\bigr)_x$ is feasible and equal to $v_x$. 

%

Suppose now that $\bx^\ast = (x_1^\ast, \dots, x_m^\ast)$ are the cutoffs in the maximal equilibrium of~\eqref{eq:univalue}.\footnote{Our definition of the spillover matrix below can be modified for the minimal equilibrium case. We instead decrease the cutoff $x_{\ell}^*$ by $t$ and define $B_{k\ell}$ to be the derivative at $t=0$ of the measure of firms below the cutoff $x_k^*$ with putative equilibrium value $v\bigl(\kappa(\,\cdot\,; \bx^\ast - t\onebb_\ell)\bigr)_x \geq v^\ast - \beta$. Theorem~\ref{t:spillovermatrix} will continue to hold.} Perturb one component, say the $\ell$th, very slightly, to form the vector $\bx^\ast + t\onebb_\ell := (x_1^\ast, \dots, x_{\ell-1}^\ast, x_\ell^\ast + t, x_{\ell+1}^\ast, \dots, x_m^\ast)$ where $t>0$ is tiny. Imagine for a moment that the vector $\bx^\ast + t\onebb_\ell$
corresponds to the cutoffs of a different equilibrium in which a slightly larger set of firms of type $\ell$ are insolvent. This has a spillover effect on firms of type $k$ through the cross-holdings that firms of type $k$ have with firms of type $\ell$. It is this spillover we need to quantify. The difficulty is that $\bx^\ast + t\onebb_\ell$ need not describe an equilibrium. Instead, we will focus on the associated putative equilibrium $v\bigl(\kappa(\,\cdot\,; \bx^\ast + t\onebb_\ell)\bigr).$
\Xomit{By increasing the value of the cutoff in the $\ell$th interval $\BT_\ell$, a slightly larger set of firms of type $\ell$ are labelled putatively insolvent. \textcolor{red}{Notation issue: cutoff is a word we introduced to describe a particular kind of equilibrium. In other words, an equilibrium induces a collection of cutoffs. Now, we are taking any collection of numbers, one for each block and interpreting them as cutoffs. Aren't we doing this: starting with cutoffs, perturbing, then, applying $\kappa$ and then looking for a putative equilibrium wrt $\kappa$?} This has a spillover effect on firms of type $k$ through the cross-holdings that firms of type $k$ have with firms of type $\ell$.
\textcolor{blue}{Let $\BI_k = \BI_k(t) = \BI_k(t;\bx^\ast)$ denote the set of $x$ in the interval $\BT_k$ for which $x\geq x_k^*$ but the corresponding value of the putative equilibrium satisfies $v\bigl(\kappa(\,\cdot\,; \bx^\ast + t\onebb_\ell)\bigr)_x < v^\ast$. In words, $\BI_k(t; \bx^\ast)$ is that subset of points in the interval $\BT_k$ which, while initially solvent, are driven into insolvency by the perturbation. This captures the notion of a spillover of a perturbation in valuations of firms of type $\ell$ onto firms of type $k$. }
}

Some firms of type $k$ that are {\em putatively} solvent before the perturbation may become insolvent after the perturbation. We introduce notation to track them. Write $\BI_k = \BI_k(t) = \BI_k(t;\bx^\ast)$ for the set of $x$ in the interval $\BT_k$ for which $x\geq x_k^*$ but the corresponding value of the putative equilibrium satisfies $v\bigl(\kappa(\,\cdot\,; \bx^\ast + t\onebb_\ell)\bigr)_x < v^\ast$. 

With this for preparation, for each $k$ and $\ell$, set
\begin{equation}\label{def:spillover}
	B_{k\ell} := \frac{d}{dt} 
		\lambda\bigl(\BI_k(t; \bx^\ast + t\onebb_\ell)\bigr)\bigg|_{t=0},
\end{equation}
where $\lambda$ stands for Lebesgue measure. Call the matrix $\bB = [B_{k\ell}]$ with entries $B_{k\ell} = B_{k\ell}(\bx^\ast)$ the \emph{spillover matrix}.

The entries of the spillover matrix measure the rate of change of (putative) solvency of firms of any given type occasioned by the failure of an infinitesimally small set of firms of any other given type. As foreshadowed in the preamble to Theorem~\ref{thm:maximaleq}, we need a stability condition to ensure that the positive spillovers from an injection of capital are limited.

\smallskip\noindent\textsc{Stability assumption:} \emph{The spectral radius of the spillover matrix $\bB$ is strictly less than one.}



\smallskip Recall from the construction of the block equity graphon~\eqref{block-graphon} that
\begin{equation*}
	T_{k\ell} = \frac{cg_{k\ell}}{g_{1\ell}s_1 + \dots + g_{m\ell}s_m}
\end{equation*}
represents the share of a firm of type $\ell$ that is held by a firm of type $k$. As in Example \ref{eg:PiecewiseLinear}, let $\bT = [T_{k\ell}]$ denote the matrix of cross-shares and $D$ be the diagonal matrix whose $k^{th}$ entry on the leading diagonal is $s_k$.

The following result gives an explicit expression for the spillover matrix, which is obtained by analyzing how firm failures cascade in the graphon. We include the proof here as it demonstrates the types of calculations driving several of our examples and results, but the reader can also safely skip to the characterizations of optimal cash infusions below.

\begin{thrm}\label{t:spillovermatrix}
    If $\bx^\ast = (x_k^\ast, 1\leq k\leq m)$ represents the cutoffs of a cutoff equilibrium, the corresponding entries, $B_{k k'} = B_{k k'}(\bx^\ast)$, of the associated spillover matrix are given by 
    $$B_{k k'} = \{\left(I-\bT D\right)^{-1} \bT\}_{kk'} \frac{\beta}{f'(x_k^*)}.$$
\end{thrm}
\begin{proof} 
\Xomit{
\textcolor{red}{First, we need to determine $v\bigl(\kappa(\,\cdot\,; \bx^\ast + t\onebb_\ell)\bigr)_x$. This will be characterized by the solution to equation \eqref{eq:univalue}. Recall, that the middle term on the right hand side of equation \eqref{eq:univalue}) is piecewise constant in each interval $\BT_k$. As in the proof of Proposition \ref{prop:block-graphon} let $A_k(t) = \sum_\ell T_{k\ell}\int_{\BT_\ell} v\bigl(\kappa(\,\cdot\,; \bx^\ast + t\onebb_\ell)\bigr)_y\,dy$ and let $\overline{e}_k=\frac{1}{s_k}\int_{\mathscr{C}_k} f_k(x)$ be the average endowment of firms in block $k$. Hence,
$$v\bigl(\kappa(\,\cdot\,; \bx^\ast + t\onebb_\ell)\bigr)_x= e_x + A_k(t)  \text{ when }x \geq x_k^* \text{ and } v\bigl(\kappa(\,\cdot\,; \bx^\ast + t\onebb_\ell)\bigr)_x= e_x + A_k(t) - \beta  \text{ and }x < x_k^*$$
where} \textcolor{blue}{should we be using $\mathscr{C}_k$?}
}
Let $\bx^\ast = (x_1^\ast, \dots, x_m^\ast)$ be the cutoffs in a cutoff equilibrium of~\eqref{eq:univalue}. Perturb one component, say the $j$th, very slightly by $t >0$, to form the vector 
$$\bx^\ast + t\onebb_j := (x_1^\ast, \dots, x_{j-1}^\ast, x_j^\ast + t, x_{j+1}^\ast, \dots, x_m^\ast).$$
We focus on the associated putative equilibrium $v\bigl(\kappa(\,\cdot\,; \bx^\ast + t\onebb_j)\bigr).$
Following the proof of Proposition \ref{prop:block-graphon} let $A_k(t|j) = \sum_\ell T_{k\ell}\int_{\BT_\ell} v\bigl(\kappa(\,\cdot\,; \bx^\ast + t\onebb_j)\bigr)_y\,dy$, then, $v_x = f_k(x) + A_k(t|j)$ for $x$ in the sub-interval $\BT_k$. Also, let $\overline{e}_k=\frac{1}{s_k}\int_{\BT_k} f_k(x)dx$ be the average endowment of firms in block $k$. Recall, $t_k:= \sum_{i=1}^k s_i$ is the partial sums of the block lengths with $t_0 := 0$ per standard convention. Then $\BT_k = [t_{k-1}, t_k).$

For $k \neq j $ and $x \in \BT_k $ we have 
$$v\bigl(\kappa(\,\cdot\,; \bx^\ast + t\onebb_j)\bigr)_x= e_x + A_k(t|j)  \text{ when }x \geq x_k^* \text{ and } v\bigl(\kappa(\,\cdot\,; \bx^\ast + t\onebb_j)\bigr)_x= e_x + A_k(t|j) - \beta  \text{ when }x < x_k^*.$$
For $k=j$ and $x \in \BT_j,$
$v\bigl(\kappa(\,\cdot\,; \bx^\ast + t\onebb_j)\bigr)_x= e_x + A_j(t|j)$ when $x \geq x^*_j + t$ and $v\bigl(\kappa(\,\cdot\,; \bx^\ast + t\onebb_j)\bigr)_x= e_x + A_j(t|j) - \beta$  when $x < x_j^* +t.$
Working as in Example~\ref{eg:PiecewiseLinear}, a recursive application of these identities shows that
\begin{align*}
 A_r(t|j)    &= \sum_{k=1}^m T_{rk} \int_{\BT_k} 
            \biggl[
                e_y + \sum_{\ell=1}^m T_{k \ell}
                \int_{\BT_\ell} v_z\,dz - \beta \onebb_{[t_{k-1}, x_k^\ast+ t\onebb_j)}(y)
            \biggr]\,dy\\
        &= \sum_{k=1}^m T_{rk} \int_{\BT_k}
            \bigl[
                e_y + A_k(t|j) - \beta \onebb_{[t_{k-1}, x_k^\ast+ t\onebb_j)}(y)
            \bigr]\,dy\\
        &= \sum_{k=1}^m T_{rk} 
            \bigl[ 
                \overline{e}_k s_k + A_k(t|j) s_k - \beta (x_k^\ast + t\onebb_j - t_{k-1})\\
                &= \sum_{k=1}^m T_{rk} s_k
            \bigl[ 
                \overline{e}_k  + A_k(t|j)  - \frac{\beta}{s_k} (x_k^\ast + t\onebb_j - t_{k-1})
            \bigr]
\end{align*}
\Xomit{\textcolor{orange}{Version 1:
Let $\Delta^1$ be the column vector whose $k^{th}$ component is $\overline{e}_k s_k$. Let $D$ be the diagonal matrix whose $k^{th}$ entry on the leading diagonal is $s_k$. Let $\Delta^2$ be the column vector whose $k^{th}$ component is $- s_k^{-1}\beta (x_k^\ast + t\onebb_j - b_{k-1})$. Using $A(t|j)$ to denote the vector whose $r^{th}$ component is $A_r(t|j)$ we can rewrite the expression in matrix-vector form:
$$A(t|j) = T \Delta^1 + T D A(t|j) + T \Delta^2$$
$$\Rightarrow A(t|j) = (I - TD)^{-1}T[\Delta^1 + \Delta^2]$$
\begin{equation}\label{xi}
    \begin{pmatrix} A_1(t|j) \\  \vdots \\ A_{k'}(t|j) \\ \vdots \\ A_m(t|j) \end{pmatrix} =\left(I-TD\right)^{-1}T \cdot \begin{pmatrix} \overline{e}_1s_1  - \beta (x_1^\ast + t\onebb_j - t_{0}) \\ \vdots \\ \overline{e}_{k'}s_{k'} - \beta (x_{k'}^\ast + t\onebb_j - t_{k'-1})\\ \vdots \\  \overline{e}_ms_m  - \beta (x_m^\ast + t\onebb_j - t_{m-1}) \end{pmatrix}.
\end{equation}
}
}

Let $\Delta$ denote the column vector whose $k^{th}$ row is $s_k^{-1}\beta (x_k^\ast + t\onebb_j - t_{k-1}).$ 
Then,
$$A(t|j) = \bT D{\bar e} + \bT DA(t|j) - \bT D\Delta$$
$$\Rightarrow A(t|j) = (I-\bT D)^{-1}\bT D({\bar e}- \Delta)$$
\begin{equation}\label{xi}
    \begin{pmatrix} A_1(t|j) \\  \vdots \\ A_{k'}(t|j) \\ \vdots \\ A_m(t|j) \end{pmatrix} =\left(I-\bT D\right)^{-1}\bT \cdot \begin{pmatrix} \overline{e}_1s_1  - \beta (x_1^\ast + t\onebb_j - t_{0}) \\ \vdots \\ \overline{e}_{k'}s_{k'} - \beta (x_{k'}^\ast + t\onebb_j - t_{k'-1})\\ \vdots \\  \overline{e}_ms_m  - \beta (x_m^\ast + t\onebb_j - t_{m-1}) \end{pmatrix}.
\end{equation}

We use equation \eqref{xi} to determine $\frac{\partial A_k(t|j)}{\partial t}$. To this end, denote the entry in the $k^{th}$ row and $r^{th}$ column of $\left(I-\bT D\right)^{-1}\bT$  by $\{\left(I-\bT D\right)^{-1}\bT\}_{k, r}.$ Then,
$$\frac{\partial A_k(t|j)}{\partial t} = -\beta \{\left(I-\bT D\right)^{-1}\bT\}_{k,j}.$$
If we perturb $x^*_{k'}$ by $t$ this reduces the value of all firms in block $k$ by $\{\left(I-\bT D\right)^{-1}\bT\}_{kk'}\beta t$. The additional firms that fail as a consequence are those whose values were between $v^*$ and $v^*+\{\left(I-\bT D\right)^{-1}\bT\}_{kk'}\beta t$ before the perturbation. Since $v_{x_{k^*}}=v^*$ and the relevant firms' endowments and values differ only by a constant, these are exactly the firms with endowments between  $f(x_k^*)$ and $f(x_k^*)+\{\left(I-\bT D\right)^{-1}\bT\}_{kk'}\beta t$. The measure of this set of firms is
$$f^{-1}(f(x_k^*)+\{\left(I-\bT D\right)^{-1}\bT\}_{kk'}\beta t)-x_k^*.$$
The entry $B_{kk'}$ is defined to be the derivative of this measure in $t$ evaluated at $t=0$. Applying the inverse function rule we determine this derivative to be
$$
\frac{d}{dt} \left.\left(f^{-1}(f(x_k^*)+\{\left(I-\bT D\right)^{-1}\bT\}_{kk'}\beta t)-x_k^*\right)\right|_{t=0}= \frac{1}{f'(x_k^*)}\cdot \{\left(I-\bT D\right)^{-1}\bT\}_{kk'}\beta
$$
as desired.
\end{proof}

\begin{example}\label{ex:oneblock}
Consider a single block with linear firm endowments given by $f(x) = ax+e^0$. Suppose the maximal equilibrium before the cash infusion is characterized by an interior cutoff $x^*$ but the budget, $K,$ is insufficient to make all firms solvent under an optimal infusion. 

To determine the spillover matrix, observe that the matrix $\bT$ in this case will be scalar with value $c$ and $D$ will just be the number 1. Hence, $\left(I-\bT D\right)^{-1}\bT = \frac{c}{1-c}$. It is also easy to see that $\frac{\beta}{f'(x_k^*)}= \frac{\beta}{a}$. Thus, the
matrix $B$  is the scalar $b=\frac{c}{1-c}\cdot \frac{\beta}{a}.$ Stability requires that $ b < 1$. 
\end{example}

\medskip\noindent\textsc{Optimal infusions}

\smallskip\noindent Given a cash infusion $\iota(x)$, let $v_x$ be the firm valuations at the maximal equilibrium with cash endowments $f_k(x)$ and $\widetilde{v}_x$ be the firm valuations at the maximal equilibrium with endowments ${f}_k(x)+\iota(x)$. Our goal is to find a cash infusion with budget $K$  that maximizes the sum  of firm valuations, $\int_x \widetilde{v}_x dx$.\footnote{Since the objective only depends on integrals of $\iota(x)$, we will characterize optimal interventions up to changes in values on measure zero sets.} This is equivalent to maximizing the fraction, $\int_x \kappa_x dx$, of solvent firms. While we focus on the maximal equilibrium our analysis extends to the minimal equilibrium as well. In this way one can bracket the magnitude of the impact of a cash infusion.

Given a cutoff equilibrium $\bx^\ast$ on the graphon, recall matrix $B$ where
$$B_{kk'} =\{\left(I-\bT D\right)^{-1}\bT\}_{kk'}\frac{\beta}{f'(x_k^*)}.$$
The entries of matrix $B$ are the instantaneous rates at which firms in block $k$ fail when an infinitesimal fraction of firms in block $k'$ fail. Finally, let $\mathbf{1}$ be the vector of ones and $E_k$ be the $k^{\text{th}}$ standard unit basis vector.

\begin{prop}\label{prop:bailout} Suppose the support of an optimal intervention $\iota(x)$ is interior in each block. In each block $k$ there exists $y_k \geq 0$ such that $\iota(\cdot)$ increases the endowments of all firms in an interval with  measure $y_k$ upto a constant $\widehat{e}_k$, but leaves the endowments of other firms in the block unchanged.
\end{prop}

The optimal policy uses global information about the network encoded in the matrix $B$ as well as local information, the endowments of each firm. The measure of firms directly rescued in each block is proportional to the Katz-Bonacich centrality of the block in the weighted network defined by the $m \times m$ matrix $B$, adjusted to account for the cost of rescuing firms in that block. 

Within each block, the optimal infusion increases the cash endowments of an interval of firms with endowments in $[\widehat{e}_k-\epsilon,\widehat{e}_k]$ to $\widehat{e}_k$. At the post-infusion maximal equilibrium, these firms all have value $v^*$. There is an additional interval of firms with endowments $e_x>\widehat{e}_k$ that would be insolvent without the infusion, but are solvent after the infusion. Intuitively, these correspond to firms that would have not failed directly but are exposed to contagion  if other firms are allowed to fail. A graphical depiction appears in Figure~\ref{f:infusion} of Example~\ref{ex:one_group}.

The next proposition explicitly characterizes the optimal intervention when the $f_k(\cdot)$ are linear. It will specify for each block, which firms should be bailed out and by how much. It will also tell us the measure of firms within a block that are {\em indirectly} bailed out because of positive spillovers.

\begin{thrm}\label{t:bailout_linear} Suppose the support of an optimal intervention $\iota(x)$ is interior in each block and each $f_k(x)$ is linear with slope $a_k$. Then, for each block $k$ there exists $\delta_k$ and $y_k$ such that the optimal intervention increases the endowments of firms in the intervals $[x_k^*-\delta_k-y_k,x_k^* - \delta_k]$ to $e_{x_k^*-\delta_k}$, where the $y_k$ are characterized by
\begin{equation}\label{eq:ratio}\frac{\beta/(a_ky_k) + 1}{\beta/ (a_{k'}y_{k'}) +1} = \frac{\mathbf{1}^{\top} (I-B)^{-1} E_{k'}}{\mathbf{1}^{\top} (I-B)^{-1} E_k
 } \end{equation}
 for all $k$ and $k'$ and the budget constraint
 $$\frac12 \sum_k a_ky_k^2 = K,$$
and the $\delta_k$  are given by $$ \delta_k = E_k^{\top} (I-{B})^{-1}B \begin{pmatrix}
    y_1 + \frac{a_1y_1^2}{2\beta}
    \\  \vdots \\ y_m + \frac{a_my_m^2}{2\beta}\end{pmatrix}.$$
\end{thrm}

The assumption that $f_k(\cdot)$ is linear can be interpreted as the endowments in each block being uniformly distributed. Nevertheless, Theorem~\ref{t:bailout_linear} can be extended to accommodate $f_k(\cdot)$ that are non-linear but approximable by a piecewise linear function with a bounded number of segments. Within each block split the interval of endowments into several pieces consisting of firms with endowments in the same linear segment and treat each of these pieces as a block.

Theorem~\ref{t:bailout_linear} characterizes the $y_k$ implicitly. The next two examples provide an illustration of how the optimal intervention can be computed from this characterization.
 
\begin{example}\label{ex:one_group} We continue Example \ref{ex:oneblock} by computing the optimal cash infusion. Recall that the
matrix $B$  is the scalar $b=\frac{c}{1-c}\cdot \frac{\beta}{a}.$ 
 
 If $y$ is the measure of firms in the block that receive an injection, then, Proposition~\ref{prop:bailout} tells us that for some $\delta >0$ the optimal intervention increases the endowments of firms in an interval $[x^*-\delta-y, x^* - \delta]$ to $e_{x^*-\delta}$. Therefore, each firm $x \in [x^*-\delta-y, x^* - \delta)$ receives an infusion of $a(x^*-\delta -x),$ which means the total infusion is $a\int_{x^*-\delta -y}^{x^*-\delta}(x^*-\delta -x)dx = 0.5ay^2.$ If $K$ is the total budget, and $2K/a \leq 1$, then, $y = \sqrt{2K/a}$, otherwise $y=1.$ Notice, the measure of firms that receive an injection is independent of either the exposure, $c$, or the distress cost $\beta$. In the remainder of this example we suppose $2K/a \leq 1$. 
 \Xomit{The infusion provides $a(x-(x^*-y-\delta))$ to each firm in this interval, so the budget constraint is
$K = ay^2/2$ and the measure of firms receiving a cash infusion is $y = \sqrt{2K/a}$.}

Consider Figure~\ref{f:infusion}. The top two panels (Figures ~\ref{f:infusion}(a) and ~\ref{f:infusion}(b)) show firm valuations (blue solid lines) and endowments (red dashed line) before and after the infusion, respectively. The support of the infusion is an interval of firms $[x^*-\delta-y,x^*-\delta]$ {\em strictly} to the left of the cutoff $x^*$. The bottom panel (Figure~\ref{f:infusion}(c)) depicts which firms receive an infusion. Each of them receives sufficient cash to achieve value $v^*$ after the intervention.
\begin{figure}
\subfigure{
  \includegraphics[width=.5\textwidth]{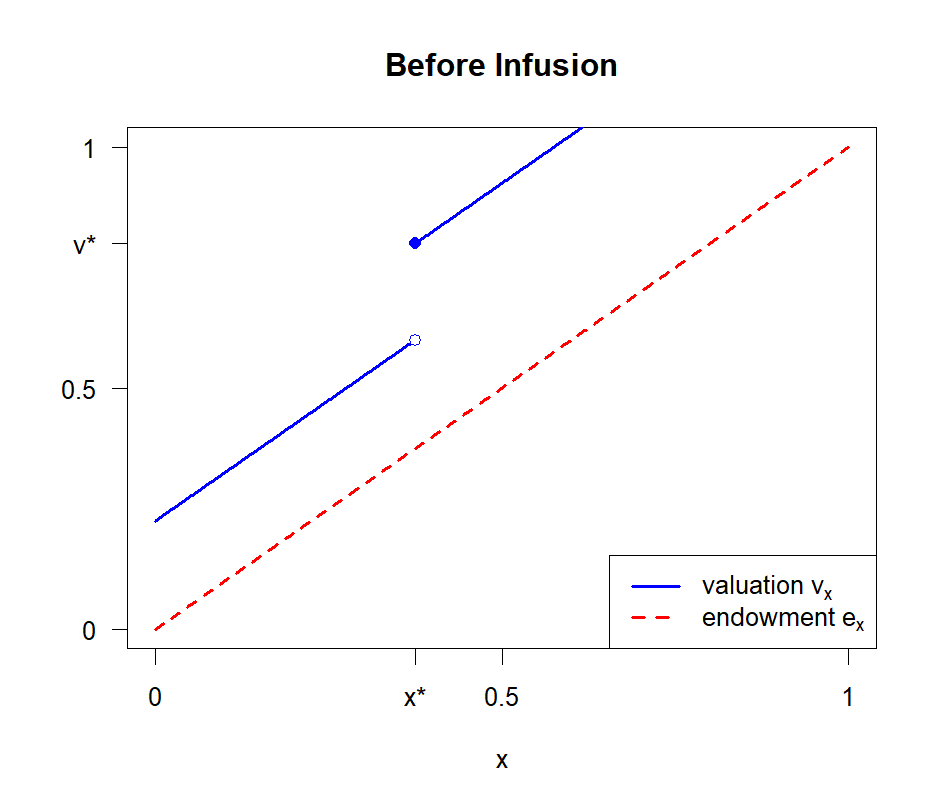}
  \label{f:infusion-a}
}
\subfigure{
  \includegraphics[width=.5\textwidth]{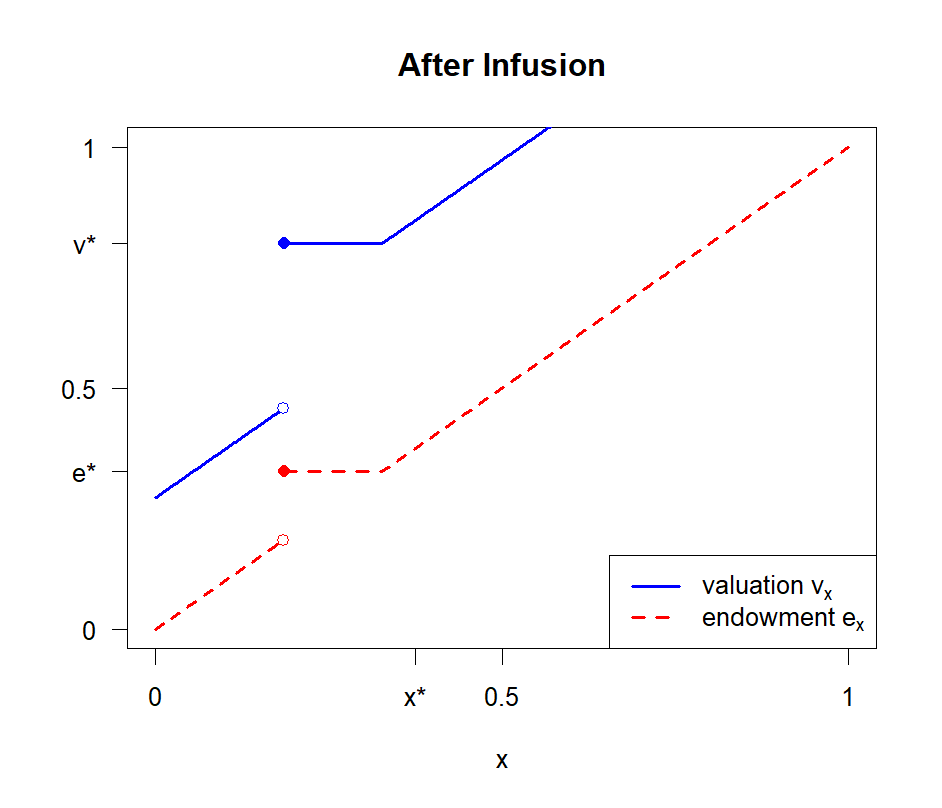}
  \label{f:infusion-b}
}
\center\subfigure{
  \includegraphics[width=.5\textwidth]{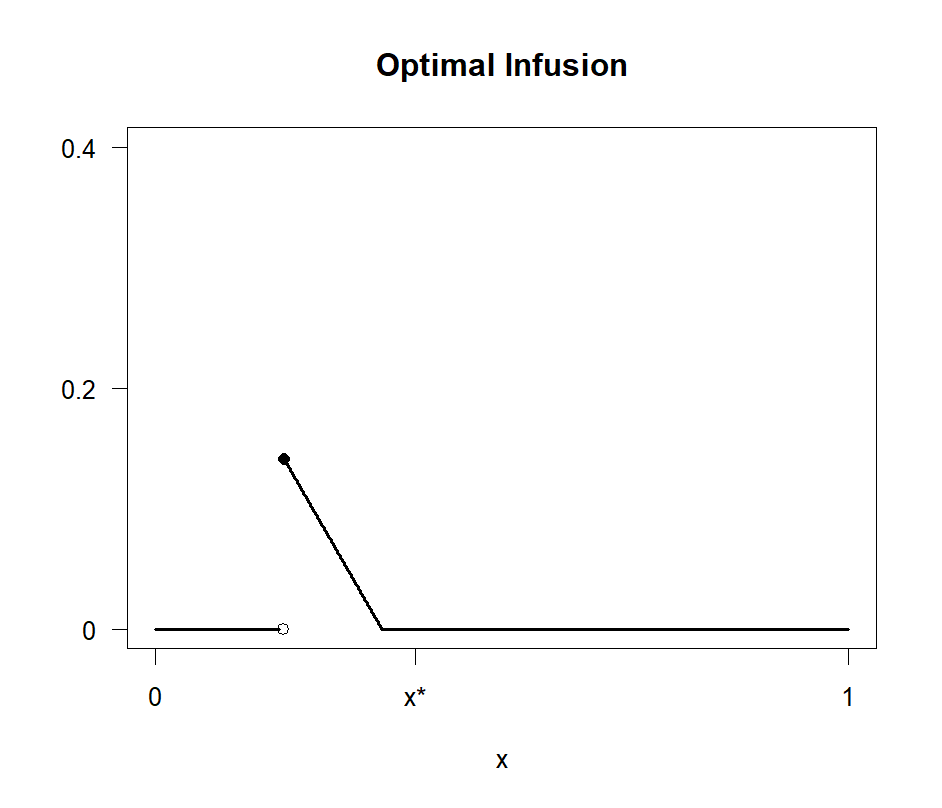}
  \label{f:infusion-c}
}
\caption{Optimal cash infusions with linear endowments and a single block.}
\label{f:infusion}
\end{figure}

\Xomit{Figure~\ref{f:infusion}(a) shows the values (solid blue line) and endowments (dashed red line) before the cash infusion. Figure~\ref{f:infusion}(b) shows the values (solid blue line) and endowments $e_x+\iota(x)$ (dashed red line) after the cash infusion. Figure~\ref{f:infusion}(c) shows the optimal infusion.}

An additional interval of firms $[x^*-\delta,x^*]$ is not directly bailed out by the cash infusion, but are nevertheless rescued due to spillovers. One interpretation is that this interval contains firms that would fail due to a domino effect in the absence of an intervention. We can, in this example, compute the measure of such firms. 

Directly moving $\epsilon$ firms from insolvency to solvency can lead to at most $\frac{\epsilon}{1-b}$ insolvent firms becoming solvent. Therefore, $\delta$, the measure of firms {\em indirectly} rescued by the optimal cash infusion is 
$$\delta = (1-b)^{-1}b(y+ 0.5ay^2/\beta)= \frac{b}{1-b}\left(\sqrt{\frac{2K}{a}} + \frac{K}{\beta}\right).$$
As one might expect, as $c$, the exposure of each firm, increases, the measure of firms that are indirectly rescued increases. This is because $b$ and $(1-b)^{-1}$ increase with $c$.

As $a$, the slope of the endowment function, increases, the measure of firms that receive an injection declines. Furthermore, because both $\frac{b}{1-b}$ and $\sqrt{2K/a}$ decline with $a$, the measure of firms that are indirectly rescued, also shrinks. In both cases it is because each firm needs less cash to get over the solvency threshold.

The number of indirectly rescued firms increases with $\beta$. This is because $\frac{b}{1-b}$ increases with $\beta$. Increasing $K$, the amount of cash available to inject into firms has three effects. First, the measure of firms that receive an injection increases. Second, the interval of firms that receive an injection shifts to the left, because the cutoff type $x^*$ becomes smaller. Thus, an insolvent firm that was close to the solvency threshold that received an injection when $K$  was small, will not necessarily receive it when $K$ increases. Third, the measure of firms indirectly rescued increases with $K$, but this indirect effect suffers diminishing returns. 
\end{example}

\begin{example}
We consider a simple core-periphery model where the core will be block 1 and the periphery will be block 2 but the cross-holdings satisfy $g_{22}=0$. Thus, firms in the periphery {\em only} own shares in core firms. Firm endowments in the core are described by $f_1(x) = a_1 x+e^0_1$ and by  $f_2(x) = a_2 x+e^0_2$ in the periphery. For economy of exposition only, assume that the core and periphery each contain half of the firms.

Before the infusion, the maximal equilibrium is characterized by two interior cutoffs,  $x_1^*$ for block 1  and $x_2^*$ for block 2.  Proposition~\ref{prop:bailout} tells us that for some $\delta_1,\delta_2>0$, the optimal intervention increases the endowments of firms in intervals $[x_k^*-\delta_k-y_k, x_k^* - \delta_k]$ to $e_{x_k^*-\delta_k}$ for $k \in \{1,2\}.$ We describe the intervention when $K$ is small enough, so that $x_k^*-\delta_k-y_k$ is interior in each block. 
The budget constraint is
$$\frac12 a_1y_1^2 + \frac12 a_2y_2^2 = K.$$

The matrix of equity cross-holdings is given by $$T = 2c  \begin{pmatrix} \frac{g_{11}}{g_{11}+g_{21}} & 1 \\ \frac{g_{21}}{g_{11}+g_{21}} & 0  \end{pmatrix},$$
while $D=\frac12 I$. By Theorem \ref{t:spillovermatrix}, the spillover matrix is
\begin{equation}\label{core-per}B = 
\frac{c \beta}{(1-c) (c g_{21} + g_{11} + g_{21})}
\begin{pmatrix}
(c g_{21} + g_{11})/a_1 & (g_{11} + g_{21})/a_1  \\
g_{21}/a_2 & c g_{21}/a_2
\end{pmatrix}.
\end{equation}
\Xomit{
\textcolor{red}{
\text{Spectral Radius} $$= \max \left( \left| \frac{-\left(\frac{(c g_{21} + g_{11})}{a_1} + \frac{(c g_{21})}{a_2}\right) + \sqrt{\left(\frac{(c g_{21} + g_{11})}{a_1} + \frac{(c g_{21})}{a_2}\right)^2 - 4 \left(\frac{(c g_{21} + g_{11})(c g_{21})}{a_1a_2} - \frac{g_{11}g_{21}}{a_1a_2}\right)}}{2} \right|, \right.
$$
$$\left| \frac{-\left(\frac{(c g_{21} + g_{11})}{a_1} + \frac{(c g_{21})}{a_2}\right) - \sqrt{\left(\frac{(c g_{21} + g_{11})}{a_1} + \frac{(c g_{21})}{a_2}\right)^2 - 4 \left(\frac{(c g_{21} + g_{11})(c g_{21})}{a_1a_2} - \frac{g_{11}g_{21}}{a_1a_2}\right)}}{2} \right|$$
The larger is clearly,
$$ \frac{\left(\frac{(c g_{21} + g_{11})}{a_1} + \frac{(c g_{21})}{a_2}\right) + \sqrt{\left(\frac{(c g_{21} + g_{11})}{a_1} + \frac{(c g_{21})}{a_2}\right)^2 - 4 \left(\frac{(c g_{21} + g_{11})(c g_{21})}{a_1a_2} - \frac{g_{11}g_{21}}{a_1a_2}\right)}}{2}
$$
$$= \frac{\left(\frac{(c g_{21} + g_{11})}{a_1} + \frac{(c g_{21})}{a_2}\right) + \sqrt{\left(\frac{(c g_{21} + g_{11})}{a_1} - \frac{(c g_{21})}{a_2}\right)^2  - \frac{g_{11}g_{21}}{a_1a_2}}}{2}
$$}
}
Applying equation \eqref{eq:ratio}, \begin{equation}\label{eq:indifference_core_periphery}
\frac{\beta /(a_1y_1)+1}{\beta/ (a_2 y_2) + 1} = \frac{a_1a_2(c g_{21} + g_{11}+ g_{21}) + a_2c \beta g_{21}}{a_1a_2(c g_{21} + g_{11}+ g_{21})+  a_1c \beta g_{21}}.\end{equation}
 When $a_1=a_2$, for example, because endowments are drawn from the same distribution in the core and periphery, the optimal bailout injects cash to the same measure of firms in each of the blocks. In general, the optimal bailout injects cash to more firms in a block where the endowments are less dispersed because there are more firms close to the failure threshold. But that difference diminishes when $g_{21}$ is larger and $g_{11}$ is smaller, i.e., when the peripheral firms are more exposed to the core firms.

The budget constraint describes an ellipse in the $(y_1, y_2)$ space. Equation \eqref{eq:indifference_core_periphery} describes a hyperbola and it is straightforward to show that there is a single intersection in the non-negative orthant.
 
\Xomit{The budget constraint is of the form 
$$\frac{a}{x}+a= \frac{b}{y}+b\,\, \Rightarrow y= \frac{bx}{a-(b-a)x}.$$
This makes it a hyperbola.
The budget constraint describes an ellipse in the $(y_1, y_2)$ space. When $a_1=a_2$, equation \eqref{eq:ratio} becomes a straight line with positive slope and the budget constraint is a circle, so there is a unique solution in non-negative orthant. In general, they could intersect at upto 4 points. Possibly two of them could be in the non-negative orthant. However, the derivative of $y$ wrt $x$ is $\frac{ba}{[a-(b-a)x]^2}$ which is positive. So, only one intersection in non-negative orthant.}
\Xomit{Recall the budget constraint, $0.5a_1y_1^2+ 0.5a_2y_2^2 = K.$
This along with the relationship between $y_1$ and $y_2$ in equation \eqref{eq:indifference_core_periphery}  pin down the sizes of the intervals $y_1$ and $y_2$.
}
The measure $\delta_k$ of firms indirectly rescued in each block is $$ \delta_k = E_k^{\top} (I-{B})^{-1}B \begin{pmatrix}
    y_1 + \frac{a_1y_1^2}{2\beta}
    \\ y_2 + \frac{a_2y_2^2}{2\beta}
\end{pmatrix}.$$

In the special case $a_1=a_2=a$, so that $y_1 = y_2$ as well, we have
$$\frac{\delta_1}{\delta_2} =1 + \frac{2a g_{11}}{(a(1+c)+\beta c)g_{21} }>1.$$
Therefore, while the number of firms directly bailed out is the same in each block in this case, more firms are indirectly snatched from insolvency in the core than in the periphery. This is because firms in the core have equity holdings in both groups rather than only the core, so they benefit more from spillovers from the bailout. The gap is larger when $g_{11}$ is larger relative to $g_{21}$ as then the periphery's holdings are smaller.

\end{example}

\medskip\noindent\textsc{Optimal Cash Infusions in Random Networks}

\smallskip\noindent

In this section we discuss how to extend the optimal cash infusion in the graphon setting to a finite instance sampled from the graphon. We show that given a cash infusion in the graphon setting, there exists a cash infusion in the corresponding finite instance keeping at least the same proportion of firms solvent for a slightly higher budget.

As before, we consider a stochastic block equity graphon $G(x,y)$ and a sequence of $\mathscr{G}(n)$ sampled from $G(x,y)$. A cash infusion $\mathbf{\iota}^{(n)}$ for a finite random network specifies the amount of cash $\iota_i$ given to each firm $i \in \{0,\frac{1}{n-1},\hdots,1\}.$ Given a cash infusion $\mathbf{\iota}^{(n)}$, we write $\widetilde{V}$ for the values of firms at the maximal equilibrium after the cash infusion.

To maintain consistency with the continuum model, suppose the budget for the finite random network cash infusion is $(K+\epsilon)n$ for some small $\epsilon >0$. Thus, the average budget per firm is $K + \epsilon$. In other words, we increase the budget per firm by an arbitrarily small amount when we switch from the graphon to the sampled finite network. We now show that given the optimal intervention on the continuum, we can find an intervention keeping at least the same fraction of firms solvent on large random networks using a slightly higher budget. To formalize this on the graphon, recall we use $\lambda$ for the Lebesgue measure.

\begin{prop}\label{prop:finite_bailout}
Let $\epsilon>0$ and consider the optimal cash infusion $\iota(x)$ with budget $K$ for the block equity graphon. There exists a sequence of cash infusions $\iota^{(n)}$ with budgets $K+\epsilon$ such that $$\liminf_n \frac{1}{n}|\{i: \widetilde{V}_i \geq v^*\}|  \geq \lambda(\{x \in (0,1]:\widetilde{v}_x \geq v^* \}).$$
\end{prop}

The construction gives a bit of additional cash to firms near the solvency threshold to insure against adverse network draws. As the network grows large, our concentration results imply the amount of additional cash needed vanishes.


\section{Conclusion}
Even if one ignores the problems of moral hazard, private information and equilibrium multiplicity, the problem of finding which firms in a financial network to bailout, is a difficult one. This paper uses a continuum analog of the underlying financial network to identify which firms should be prioritized for a bailout. Whether a firm is bailed out or not depends upon its cash endowment and how `central' is the block it belongs to. The number of firms directly bailed out in each block is proportional to the block's Katz-Bonacich centrality, adjusted to account for the cost of rescuing firms in that block.

\section*{Acknowledgements}
We thank Mengjia Xia for useful comments and Chengyang Zhu for useful comments and implementing numerical examples.
\bibliographystyle{ecta}
\bibliography{ER}

\appendix
\section{Omitted Proofs}

\begin{proof}[Proof of Theorem~\ref{thm:fixedkappa}]
We will prove a more general result allowing link probabilities $g_{k\ell}$ to vanish at a polynomial rate.

Given a matrix $B$, we let $\| B \|_2 = \max_{\|x\|=1} \|Bx\|_2$ be the matrix 2-norm. We omit the indices $n$ of matrices in the proof.

Define $\Lambda = \bC - \overline{\bC}$ to be the difference between the matrix of cross-holdings and its expectation. We begin with a lemma bounding this difference.

\begin{lemma}\label{lem:matrixnorm}
Let $p = \min_{\substack{k,\ell \\ g_{k\ell} > 0}} g_{k\ell}$. Given $\alpha<1$ and  $\zeta < 1-\alpha$, there exists $\gamma>0$ such that for $n$ sufficiently large
$$\|\boldsymbol{\Lambda}\|_2\leq  \sqrt{\frac{1}{pn^{\alpha}}}$$
with probability at least $1-\exp(-\gamma pn^{\zeta})$.
\end{lemma}
\begin{proof}

We first observe that:
$$\|\boldsymbol{\Lambda}\|_{2}=\sqrt{\left\|\boldsymbol{\Lambda}^{\top}\mathbf{\Lambda}\right\|_{2}}=\sqrt{\left\|\sum_{i=1}^{n} \boldsymbol{\Lambda}_{* i} \mathbf{\Lambda}_{* i}^{\top}\right\|_{2}}.$$

We will bound the norm $\left\|\sum_{i=1}^{n} \boldsymbol{\Lambda}_{* i} \mathbf{\Lambda}_{* i}^{\top}\right\|_{2}$ using Bernstein's inequality:
\begin{lemma}\label{bern}[Bernstein's Inequality for Matrices] Suppose $\mathbf{Z}_{1}, \ldots, \mathbf{Z}_{n}$ is a sequence of independent $d \times d$ Hermitian matrices satisfying
$$
\mathbf{E}\left(\mathbf{Z}_{i}\right)={0} \text { and }\left\|\mathbf{Z}_{i}\right\|_{2} \leq \mathcal{L} \quad(1 \leq i \leq n)
$$
Write $\mathbf{Y}=\sum_{i} \mathbf{Z}_{i}$ and let $\mathcal{V}(\mathbf{Y})$ denote the matrix variance statistic of the sum:
$$
\mathcal{V}(\mathbf{Y})=\left\|\mathbf{E}\left(\mathbf{Y}^{2}\right)\right\|_{2}=\left\|\sum_{i} \mathbf{E}\left(\mathbf{Z}_{i}^{2}\right)\right\|_{2}
$$
Then, for all $t \geq 0$,

$$
\mathbf{P}\left[\|\mathbf{Y}\|_{2} \geq t\right] \leq 2 d \exp \left(\frac{-t^{2} / 2}{\mathcal{V}(\mathbf{Y})+\mathcal{L} t / 3}\right) \leq 2 d \exp \left(-\frac{1}{4} \min \left\{\frac{t^{2}}{\mathcal{V}(\mathbf{Y})}, \frac{t}{\mathcal{L}}\right\}\right)
.$$\end{lemma}
\medskip

Fix $\rho \in (0,1)$. To apply Lemma \ref{bern}, condition on the event $\mathcal{Q}$ that for all nodes $i$ and all blocks $\ell$, the number of links $d_{\ell i}$ are in the interval $$[g_{\ell k}S_{\ell} (1-n^{-\rho}), g_{\ell k}S_{\ell} (1+n^{-\rho})]$$
 (where $k$ is the block containing node $i$). The event holds when all $d_{\ell i}$ are sufficiently close to their expected values. By the Chernoff bounds, the probability of the complement of this event vanishes at an exponential rate in $(pn)^{\rho}$.

We define the matrix $B^{(i)}$ by $B^{(i)}_{jk} = \mathbf{E}[(\boldsymbol{\Lambda}_{* i} \mathbf{\Lambda}_{* i}^{\top})_{jk}|\mathcal{Q}]$ and let $\mathbf{Z}^{(i)} = \boldsymbol{\Lambda}_{* i} \boldsymbol{\Lambda}_{* i}^{\top}- B^{(i)}$. We will:

\begin{enumerate}
\item[(1)] Bound the spectral norm of the entry-wise mean $\|B^{(i)}\|_2$.
\item[(2)] Bound the spectral norm of the variance $\|\sum_i\mathbf{E}[(\mathbf{Z}^{(i)})^2|\mathcal{Q}]\|_2$.
\item[(3)] Bound the spectral norm  $\|\mathbf{Z}^{(i)}\|_2$ uniformly on the event $\mathcal{Q}$.
\end{enumerate}

We now show three claims. We will use the following growth rates repeatedly: is $S_{\ell}$  grows at the same rate as $n$ for all $\ell$ because each block size is non-vanishing, when the event $\mathcal{Q}$ holds $\frac{g_{\ell k}}{\overline{d}_i}$ is $O(\frac{1}{n})$ for all blocks $\ell$ and $k$ and nodes $i$ in block $k$, and  when the event $\mathcal{Q}$ holds  $\frac{1}{d_i}$ and $\frac{1}{\overline{d}_i}$ are $O(\frac{1}{pn})$ for all $i$.

Claim (1): $\|B^{(i)}\|_2$ is $O(\frac{1}{pn^{1+\rho}}).$

The entries $\boldsymbol{\Lambda}_{ij}$ for $i$ in block $k$ and $j$ in block $\ell$ are equal to $\frac{c}{d_j}-\frac{cg_{k\ell}}{\overline{d}_j}$ if $A_{ij}=1$ and $-\frac{cg_{k\ell}}{\overline{d}_j}$ otherwise. So if $i$ is in block $k$ and $j$ is in block $\ell$, we compute that
\begin{align*}B^{(i)}_{jj}& =\mathbf{E}\left[\boldsymbol{\Lambda}_{j i}^2\middle|\mathcal{Q}\right]
\\ & =  c^2\mathbf{E}_{d_i}\left[\frac{1}{S_{\ell}} \left(\left(S_{\ell}-d_{\ell i}\right) \frac{g_{\ell k}^2}{\overline{d}_i^2} + d_{\ell i} \left(\frac{1}{d_i}-\frac{ g_{\ell k}}{\overline{d}_i}\right)^2\right)\middle|\mathcal{Q}\right].
\end{align*}
The right-hand side is $O(\frac{1}{pn^2}).$

Now consider nodes $i$ in block $k$, $j$ in block $\ell$, and $j'$ in block $\ell'$ with $j\neq j'$. We compute that for $\ell\neq \ell'$,
\begin{multline*}
    B^{(i)}_{jj'} =\mathbf{E}\left[\boldsymbol{\Lambda}_{j i}     
        \boldsymbol{\Lambda}_{j' i}|\mathcal{Q}\right]
    = c^2\mathbf{E}_{d_i}
        \biggl[
            \frac{1}{S_{\ell}S_{\ell'}} \biggl(d_{\ell i}\biggl(\frac{1}{d_i} - \frac{g_{\ell k}}{\overline{d}_i}\biggr) \biggl(\frac{d_{\ell'i}}{d_i}-\frac{ S_{\ell'}  g_{\ell' k}}{\overline{d}_i}\biggr)\\
            + \biggl(S_\ell - d_{\ell i} \biggr) \biggl(- \frac{g_{\ell' k}}{\overline{d}_i}  \biggr) \biggl(\frac{d_{\ell'i}}{d_i}-\frac{S_{\ell'}  g_{\ell' k}}{\overline{d}_i}\biggr)\biggr)
        \biggm| \mathcal{Q}\biggr].
\end{multline*}
Since we have conditioned on the event $\mathcal{Q}$, the quantity $$\frac{d_{\ell'i}}{d_i}-\frac{S_{\ell'}  g_{\ell' k}}{\overline{d}_i}= \frac{(\overline{d}_i - d_i)d_{\ell' i} + d_i (d_{\ell'i} - S_{\ell'} g_{\ell' k})}{d_i \overline{d}_i}$$ is $O(n^{-\rho})$ because $\overline{d}_i-d_i$ and $d_{\ell'i} - S_{\ell'}g_{\ell'k}$ are $O(pn^{1-\rho})$. So $B_{jj'}$ is $O(\frac{1}{n^{2+\rho}}).$

Similarly, for $\ell = \ell'$ we compute that
\begin{align*}B^{(i)}_{jj'}& =\mathbf{E}\left[\boldsymbol{\Lambda}_{j i} \boldsymbol{\Lambda}_{j' i}|\mathcal{Q}\right]
\\ & = c^2\mathbf{E}_{d_i}\left[\frac{1}{S_{\ell}(S_{\ell'}-1)} \left(d_{\ell i}\left(\frac{1}{d_i} - \frac{g_{\ell k}}{\overline{d}_i}\right) \left(\frac{d_{\ell i}-1}{d_i}-\frac{ (S_{\ell} -1) g_{\ell k}}{\overline{d}_i}\right)  \right. \right. \\ &  \left. \left. + \left(S_\ell - d_{\ell i} \right) \left(- \frac{g_{\ell k}}{\overline{d}_i}  \right) \left(\frac{d_{\ell i}-1}{d_i}-\frac{(S_{\ell}  -1)g_{\ell k}}{\overline{d}_i}\right)\right)\biggm| \mathcal{Q} \right]. \end{align*}
The right-hand side is again $O(\frac{1}{n^{2+\rho}})$ when the event $\mathcal{Q}$ holds.

From the bounds on the entries of $B^{(i)}$, we can express $B^{(i)}$ as the sum of a diagonal matrix with entries $O(\frac{1}{pn^2})$ and an arbitrary matrix with entries $O(\frac{1}{n^{2+\rho}})$. Both of these summands have 2-norm at most $O(\frac{1}{pn^{1+\rho}})$. Applying the triangle inequality to this sum, we can conclude that $\|B^{(i)}\|_2$ is $O(\frac{1}{pn^{1+\rho}}).$

Claim (2): $\|\sum_i\mathbf{E}[(\mathbf{Z}^{(i)})^2|\mathcal{Q}]\|_2$ is $O(\frac{1}{p^2n^{1+\rho}})$.

By the triangle inequality, it is sufficient to show that $\|\mathbf{E}[(\mathbf{Z}^{(i)})^2|\mathcal{Q}]\|_2$ is $O(\frac{1}{p^2n^{2+\rho}})$ for all $i$. We have
$$\|\mathbf{E}[(\mathbf{Z}^{(i)})^2|\mathcal{Q}]\|_2 \leq \|\mathbf{E}\left[\boldsymbol{\Lambda}_{* i} \boldsymbol{\Lambda}_{* i}^{\top}\boldsymbol{\Lambda}_{* i} \boldsymbol{\Lambda}_{*i}^{\top}|\mathcal{Q}\right]\|_2  + \|(B^{(i)})^2\|_2.$$

We showed in the previous claim that $\|B^{(i)}\|_2$ is $O(\frac{1}{pn^{1+\rho}}),$ so $\|(B^{(i)})^2\|_2$ is $O(\frac{1}{p^2n^{2+2\rho}})$ by the submultiplicativity of the spectral norm. Relaxing the upper bound, this term is $O(\frac{1}{p^2n^{1+\rho}})$ since $\rho>0$.

To complete the proof of the claim, we must bound $\|\mathbf{E}\left[\boldsymbol{\Lambda}_{* i} \boldsymbol{\Lambda}_{* i}^{\top}\boldsymbol{\Lambda}_{* i} \boldsymbol{\Lambda}_{*i}^{\top}|\mathcal{Q}\right]\|_2$. Say $i$ is in block $k$. The inner product is\begin{equation}\label{eq:innerproduct}\boldsymbol{\Lambda}_{* i}^{\top}\boldsymbol{\Lambda}_{* i} = c^2 \sum_{\ell} \left( d_{\ell i}  \left(\frac{1}{d_i} - \frac{g_{\ell k}}{\overline{d}_i}\right)^2 + (S_{\ell}-d_{\ell i}) \left(-\frac{g_{\ell k}}{\overline{d}_i}\right)^2\right).\end{equation}
So we have
$$\mathbf{E}\left[\boldsymbol{\Lambda}_{* i} \boldsymbol{\Lambda}_{* i}^{\top}\boldsymbol{\Lambda}_{* i} \boldsymbol{\Lambda}_{* i}^{\top}|\mathcal{Q}\right]   = c^2\mathbf{E}\left[\left(\sum_{\ell} \left( d_{\ell i}  \left(\frac{1}{d_i} - \frac{g_{\ell k}}{\overline{d}_i}\right)^2 + (S_{\ell}-d_{\ell i}) \left(-\frac{g_{\ell k}}{\overline{d}_i}\right)^2\right)\right) \boldsymbol{\Lambda}_{* i}  \boldsymbol{\Lambda}_{* i}^{\top}\middle|\mathcal{Q}\right].$$

We computed in the previous claim that conditioning on the event $\mathcal{Q}$, the expectations of the diagonal entries of $\boldsymbol{\Lambda}_{* i}  \boldsymbol{\Lambda}_{* i}^{\top}$ are $O(\frac{1}{pn^{2}})$ and the expectations of the off-diagonal entries are $O(\frac{1}{n^{2+\rho}})$. From equation \eqref{eq:innerproduct}, the inner product $\boldsymbol{\Lambda}_{* i}^{\top}\boldsymbol{\Lambda}_{* i}$ is $O(\frac{1}{pn})$ when $\mathcal{Q}$ holds.
So conditional on the event $\mathcal{Q}$, the diagonal entries of
$$c^2\mathbf{E}\left[\left(\sum_{\ell} \left( d_{\ell i}  \left(\frac{1}{d_i} - \frac{g_{\ell k}}{\overline{d}_i}\right)^2 + (S_{\ell}-d_{\ell i}) \left(-\frac{g_{\ell k}}{\overline{d}_i}\right)^2\right)\right) \boldsymbol{\Lambda}_{* i}  \boldsymbol{\Lambda}_{* i}^{\top}\middle|\mathcal{Q}\right].$$are $O(\frac{1}{p^2n^{3}})$ and the off-diagonal entries are $O(\frac{1}{pn^{3+\rho}})$. Therefore $\|\mathbf{E}\left[\boldsymbol{\Lambda}_{* i} \boldsymbol{\Lambda}_{* i}^{\top}\boldsymbol{\Lambda}_{* i} \boldsymbol{\Lambda}_{*i}^{\top}|\mathcal{Q}\right]\|_2$ is $O(\frac{1}{p^2n^{1+\rho}})$.

Now recall that
$$\|\mathbf{E}[(\mathbf{Z}^{(i)})^2|\mathcal{Q}]\|_2 \leq \|\mathbf{E}\left[\boldsymbol{\Lambda}_{* i} \boldsymbol{\Lambda}_{* i}^{\top}\boldsymbol{\Lambda}_{* i} \boldsymbol{\Lambda}_{*i}^{\top}|\mathcal{Q}\right]\|_2  + \|(B^{(i)})^2\|_2.$$
We have shown both terms on the right-hand side are $O(\frac{1}{p^2n^{1+\rho}})$, which proves the claim.

Claim (3): There exists  $M>0$ such that $\|\mathbf{Z}^{(i)}\|_2 \leq \frac{M}{pn}$ when $\mathcal{Q}$ holds.

\begin{align*}
\left\|\mathbf{Z}^{(i)}\right\|_{2}=\left\|\mathbf{\Lambda}_{* i} \mathbf{\Lambda}_{* i}^{\top}-B^{(i)}\right\|_{2} & \leq\left\|\boldsymbol{\Lambda}_{* i} \mathbf{\Lambda}_{* i}^{\top}\right\|_{2}+\left\|B^{(i)}\right\|_{2} \\
&\leq\left\|\boldsymbol{\Lambda}_{* i}\right\|_{2}^{2}+\left\|B^{(i)}\right\|_{2} 
\\ & = c^2 \sum_{\ell} \left( d_{\ell i}  \left(\frac{1}{d_i} - \frac{g_{\ell k}}{\overline{d}_i}\right)^2 + (S_{\ell}-d_{\ell i}) \left(-\frac{g_{\ell k}}{\overline{d}_i}\right)^2\right)+\left\|B^{(i)}\right\|_{2} .
\end{align*}
The summation is $O(\frac{1}{pn})$. Claim (1) showed that $\left\|B^{(i)}\right\|_{2}$ is $O(\frac{1}{pn^{1+\rho}})$. So $\left\|\mathbf{Z}_{i}\right\|_{2}$ is $O(\frac{1}{pn})$  when $\mathcal{Q}$ holds.

This completes the proofs of the claims. We now apply Bernstein's inequality with $d=n$. For all $t \geq 0$,
$$\mathbf{P}\left[\left\|\sum_{i=1}^n \mathbf{Z}^{(i)} \right\|_2 \geq t\middle| \mathcal{Q}\right] \leq 2n  \exp \left(-\frac{1}{4} \min \left\{\frac{t^{2}}{\mathcal{V}(\sum_{i=1}^n \mathbf{Z}^{(i)} |\mathcal{Q})}, \frac{t}{\mathcal{L}}\right\}\right).$$
We let $t=\frac{1}{2pn^\alpha}$ and will compute the minimum on the right-hand side. By Claim (2), there exists $M'>0$ such that $$\mathcal{V}\left(\sum_{i=1}^n\mathbf{Z}^{(i)} \middle|\mathcal{Q}\right) \leq \frac{M'}{p^2n^{1+\rho}}.$$
Therefore, $$\frac{t^2}{\mathcal{V}\left(\sum_{i=1}^n\mathbf{Z}^{(i)} |\mathcal{Q}\right)} \geq \frac{n^{1+\rho-2\alpha}}{ 4 M'}.$$
By Claim (3), there exists $M''>0$ such that
$$\mathcal{L} \leq \frac{M''}{pn}.$$
Therefore,
$$\frac{t}{\mathcal{L}}\geq \frac{ n^{1-\alpha}}{2 M''}.$$

Since we can choose $\rho>\alpha$, these bounds imply there exists $\gamma>0$ such that
$$\mathbf{P}\left[\left\|\sum_{i=1}^n \mathbf{Z}^{(i)} \right\|_2 \geq \frac{1}{2pn^{\alpha}} \middle| \mathcal{Q}\right] \leq 2n  \exp \left(- \gamma n^{1-\alpha} \right)$$
for $n$ sufficiently large.

We now compute that for $n$ large
\begin{align*} \mathbf{P}\left[\left\|\sum_{i=1}^n \mathbf{Z}^{(i)} \right\|_2 \geq \frac{1}{2 pn^{\alpha}} \middle| \mathcal{Q}\right]  & = \mathbf{P}\left[ \left\| \sum_{i=1}^n  \boldsymbol{\Lambda}_{*i} \boldsymbol{\Lambda}_{*i}^{\top}   + nB^{(i)}\right\|_2  \geq\frac{1}{2pn^{\alpha}}\middle|\mathcal{Q}\right]
\\ & \geq  \mathbf{P}\left[ \left\| \sum_{i=1}^n \boldsymbol{\Lambda}_{*i} \boldsymbol{\Lambda}_{*i}^{\top} \right\|_2  \geq \frac{1}{pn^{\alpha}}\middle| \mathcal{Q}\right]
\\ & = \mathbf{P}\left[\|\boldsymbol{\Lambda} \|_2  \geq \sqrt{\frac{1}{pn^{\alpha}}} \middle| \mathcal{Q}\right],
\end{align*}
where the inequality follows from Claim (1) and the triangle inequality, continuing to take $\rho>\alpha$. We conclude that $$\mathbf{P}\left[\|\boldsymbol{\Lambda} \|_2  \geq \sqrt{\frac{1}{pn^{\alpha}}} \middle| \mathcal{Q}\right]\leq 2n  \exp \left(-\gamma n^{1-\alpha} \right)$$
for $n$ large. We have shown the probability of the complement of the event $\mathcal{Q}$ vanishes at an exponential rate in $(pn)^{\rho}$, and increasing $\rho$ if necessary we can take $\rho >\zeta$. Since $\zeta<1-\alpha$ by assumption, this proves the lemma.
\end{proof}

By the preceding lemma, the Neumann series
$$(I-\boldsymbol{\Lambda})^{-1}=\sum_{k=0}^{\infty} \boldsymbol{\Lambda}^k$$
converges with probability at least $1-\exp(-\gamma pn^{\zeta})$ for $n$ sufficiently large. Since $\|\overline{\bC}\|_2=c<1$ and therefore $\|(I-\overline{\bC})^{-1}\|_2$ is bounded, the Neumann series $$(I-(I-\overline{\bC})^{-1}\boldsymbol{\Lambda}))^{-1}=\sum_{k=0}^{\infty} ((I-\overline{\bC})^{-1}\boldsymbol{\Lambda})^{k}$$
also converges with probability at least $$1-\exp(-\gamma pn^{\zeta})$$ for $n$ sufficiently large.

As in the proof of Lemma 19 of \cite*{amelkin2021}, we have
$$
\|{V}(\boldsymbol{\kappa}) - \overline{V}(\boldsymbol{\kappa}) \|_{\infty} = \| (I-(I-\overline{\bC})^{-1}\boldsymbol{\Lambda})^{-1}(I-\overline{\bC})^{-1}\boldsymbol{\Lambda} \overline{V}(\boldsymbol{\kappa})\|_{\infty}.$$
So
\begin{equation*}
\|{V}(\boldsymbol{\kappa}) - \overline{V}(\boldsymbol{\kappa}) \|_{\infty}   \leq \| (I-\overline{\bC})^{-1}\boldsymbol{\Lambda} \overline{V}(\boldsymbol{\kappa})\|_{\infty} +  \| \sum_{k=1}^{\infty} ((I-\overline{\bC})^{-1}\boldsymbol{\Lambda})^{k}(I-\overline{\bC})^{-1}\boldsymbol{\Lambda} \overline{V}(\boldsymbol{\kappa})\|_{\infty}\end{equation*}by the triangle inequality. We conclude that
\begin{equation}\|{V}(\boldsymbol{\kappa}) - \overline{V}(\boldsymbol{\kappa}) \|_{\infty}   \leq  \| (I-\overline{\bC})^{-1}\boldsymbol{\Lambda} \overline{V}(\boldsymbol{\kappa})\|_{\infty} +  \| \sum_{k=1}^{\infty} ((I-\overline{\bC})^{-1}\boldsymbol{\Lambda})^{k}(I-\overline{\bC})^{-1}\boldsymbol{\Lambda} \overline{V}(\boldsymbol{\kappa})\|_{2}\label{eq:boundV}
\end{equation}
by the comparison of the $\|\cdot\|_{\infty}$ and $\|\cdot\|_2$ norms.

\begin{lemma}
$$\left\| (I-\overline{\bC})^{-1}\boldsymbol{\Lambda} \overline{V}(\boldsymbol{\kappa})\right\|_{\infty}  \rightarrow 0$$  almost surely.
\end{lemma}
\begin{proof}

We first compute  $\boldsymbol{\Lambda} \overline{V}(\boldsymbol{\kappa})$. By symmetry, $\overline{V}(\boldsymbol{\kappa})-\mathbf{e}$ has two entries for each block $k$ corresponding to firms in that block with $\kappa_j = 0$ and $\kappa_j=1$. We call these $v_{k,0}$ and $v_{k,1}$, respectively. Then given $i$ in block $\ell$, $$(\boldsymbol{\Lambda} \overline{V}(\boldsymbol{\kappa}))_i = \sum_k \left(\sum_{\substack{j \text{ in block }k \\ \kappa_j=0}}c\left( \frac{A_{ij} }{d_j} - \frac{ g_{k \ell}}{\overline{d}_j} \right)(v_{k,0} +e_j) + \sum_{\substack{j \text{ in block }k \\ \kappa_j=1}}c\left( \frac{A_{ij} }{d_j} -\frac{ g_{k \ell}}{\overline{d}_j} \right) (v_{k,1}+e_j)\right) .$$
By the Lyapunov central limit theorem for triangular arrays, we can choose $h(n)\rightarrow 0$ such that the probability that the right-hand side has absolute value at least $h(n)$ for some node $i$ vanishes at an exponential rate in $pn$.

Next, observe that the entries $(I-\overline{\bC})^{-1}_{ij}$ depend only on the blocks of $i$ and $j$. For each node $i$ and each block $\ell$, we have
\begin{align*}
\sum_{j \text{ in block }\ell}(I-\overline{\bC})^{-1}_{ij} & = \sum_{j \text{ in block }\ell}\sum_{k=0}^{\infty}(\overline{\bC}^k)_{ij} \\ & \leq \eta \cdot \frac{1}{1-c}\end{align*}
for some constant $\eta>0$.

The $i^{th}$ entry of $ (I-\overline{\bC})^{-1}\boldsymbol{\Lambda} \overline{V}(\boldsymbol{\kappa})$ is equal to $$\sum_j (I-\overline{\bC})^{-1}_{ij} \left(\boldsymbol{\Lambda} \overline{V}(\boldsymbol{\kappa})\right)_j \leq \sum_{\text{ blocks }\ell} \eta\cdot \frac{1}{1-c} \cdot \max_{j \text{ in block }\ell} (\boldsymbol{\Lambda} \overline{V}(\boldsymbol{\kappa}))_{j}.$$
By the bound on $(\boldsymbol{\Lambda} \overline{V}(\boldsymbol{\kappa}))_{j}$ above, the probability that each summand on the right-hand side has absolute value at least $\frac{\eta h(n)}{1-c}$ vanishes at an exponential rate in $pn$. Since there are $m$ summands, this implies almost sure convergence.
\end{proof}

\begin{lemma}
$$ \left\| \sum_{k=1}^{\infty} ((I-\overline{\bC})^{-1}\boldsymbol{\Lambda})^{k}(I-\overline{\bC})^{-1}\boldsymbol{\Lambda} \overline{V}(\boldsymbol{\kappa})\right\|_{2}\rightarrow 0$$  almost surely.
\end{lemma}
\begin{proof}
By the submultiplicativity of the $\|\cdot\|_2$ norm,		
$$ \| \sum_{k=1}^{\infty} ((I-\overline{\bC})^{-1}\boldsymbol{\Lambda})^{k}(I-\overline{\bC})^{-1}\boldsymbol{\Lambda} \overline{V}(\boldsymbol{\kappa})\|_{2}\leq  \| \sum_{k=1}^{\infty} ((I-\overline{\bC})^{-1}\boldsymbol{\Lambda})^{k} \|_2 \|(I-\overline{\bC})^{-1}\|_2\|\boldsymbol{\Lambda}\|_2\| \overline{V}(\boldsymbol{\kappa})\|_{2}.$$

By Lemma \ref{lem:matrixnorm},  given $\alpha <1$ and $\zeta<1-\alpha$ we have $\|\boldsymbol{\Lambda}\|_2 \leq  \sqrt{\frac{1 }{pn^{\alpha}}}$ with probability at least $1-\exp(-\gamma pn^{\zeta})$. Because $\|\overline{C}\|_2=c < 1$, the norm $\|(I-\overline{C})^{-1}\|_2$ is bounded by some constant $M'>0$. Finally, the norm $\| \overline{V}(\boldsymbol{\kappa})\|_{2}$ is at most $M'' \sqrt{n}$ for some constant $M''>0$ because the endowments $\mathbf{e}$ are bounded. Combining these observations, we find that
$$\left\| \sum_{k=1}^{\infty} ((I-\overline{C})^{-1}\boldsymbol{\Lambda})^{k}\right\|_2 \|(I-\overline{C})^{-1}\|_2\|\boldsymbol{\Lambda}\|_2\| \overline{V}(\boldsymbol{\kappa})\|_{2} \leq \sum_{k=1}^{\infty} \left(M'  \sqrt{\frac{1}{pn^{\alpha}}} \right)^k   \left(  M'\sqrt{\frac{1}{pn^{\alpha}}}\right) M'' \sqrt{n}$$
with probability at least $1-\exp(-\gamma pn^{\zeta})$. Whenever $p$ grows at rate at least $n^{-a}$ for some $a<\frac12$, we can choose $\alpha<1$ such that the right-hand side converges to zero as $n\rightarrow \infty$. Taking any positive $\zeta<1-\alpha$, this shows almost sure convergence.
\end{proof}

Applying the two previous lemmas to inequality \eqref{eq:boundV} completes the proof.
\end{proof}
\begin{proof}[Proof of Proposition \ref{prop:Tarski}]The equilibria of the equity graphon $C$ are the fixed points (if any exist) of the map $\tau\colon v\mapsto w$ where
\begin{equation*}
	w_x = (\tau v)_x := e_x + \int_0^1 C(x,y) v_y\,dy - \beta\onebb_{v_x<v^*}.
\end{equation*}
Begin with the easy observation that $\tau$ is order-preserving. Suppose indeed that $u\leq v$. Then $\onebb_{u_x < v^\ast} \geq \onebb_{v_x < v^\ast}$ and, as $C(x,y)$ is non-negative, $\int_0^1 C(x,y) u_y\,dy \leq \int_0^1 C(x,y) v_y\,dy$. We conclude that
\begin{equation*}
	(\tau u)_x = e_x + \int_0^1 C(x,y) u_y\,dy - \beta\onebb_{u_x < v^\ast}
		\leq e_x + \int_0^1 C(x,y) v_y\,dy - \beta\onebb_{v_x < v^\ast}
		= (\tau v)_x
\end{equation*}
for each $x$, whence $\tau u\leq\tau v$.

Moreover, if $v$ lies in a sufficiently large bounded box, then so does $w = \tau v$, or, in notation, if, for some sufficiently large $D$, $|v|\leq D$, then $|w|\leq D$. Indeed, there exists $M$ such that $|e_x|\leq M$ as the endowment function is bounded and, bearing in mind the sub-stochasticity condition on the equity graphon $C$, we see that
\begin{equation*}
	|(\tau v)_x| \leq M + c\int_0^1 |v_y|\,dy + \beta \leq M + cD + \beta.
\end{equation*}
The right-hand side is $\leq D$ for any selection $D\geq (M + \beta)/(1 - c)$.

With $D$ so chosen, consider the space $\sF_D$ of measurable, bounded functions $v\colon (0,1]\to\real$ with $|v|\leq D$. Restricting $\tau$ to $\sF_D$, we conclude that $\tau\big|_{\sF_D}$ is order-preserving and maps each element of $\sF_D$ back into $\sF_D$. The stage is set for the Knaster--Tarski theorem~\citep[Theorem~1, (i-iii)]{tarski1955lattice}: as $\sF_D$ is a non-empty, complete lattice with respect to the pointwise partial order $\leq$, the set of fixed points of $\tau\big|_{\sF_D}$ forms a non-empty, complete lattice. \end{proof}

\begin{proof}[Proof of Proposition~\ref{prop:block-graphon}] The proof is by contradiction and follows by a small (and simpler) variant of the swap construction of Example~\ref{eg:swap}. Begin with an extremal equilibrium $v$.

(i) Suppose that for some $k$, there exist $t>0$ and points $x$, $x+t\in\BT_k$ with $v_x > v_{x+t}$. By~\eqref{eq:univalue},
\begin{gather*}
	v_x = f_k(x) + \sum\nolimits_\ell T_{k\ell} \int_{\BT_\ell} v_y\,dy 
		- \beta\onebb_{v_x < v^\ast},\\
	v_{x+t} = f_k(x+t) + \sum\nolimits_\ell T_{k\ell} \int_{\BT_\ell} v_y\,dy 
		- \beta\onebb_{v_{x+t} < v^\ast}.
\end{gather*}
As $f_k(x)<f_k(x+t)$ and the term $ \sum_{\ell=1}^m T_{k\ell} \int_{\BT_\ell} v_y\,dy$ is the same in both expressions, it must be the case that $v_x \geq v^\ast$ (whence $\onebb_{v_x < v^\ast} = 0$) while $v_{x+t} < v^\ast$ (whence $\onebb_{v_{x+t} < v^\ast} = 1$).

If $v$ is maximal, consider the effect of changing $v_{x+t}$ to $v'_{x+t} = v_{x+t} + \beta$ and keeping all other values the same, $v'_\xi = v_\xi$ if $\xi\neq x+t$. 
Then
\begin{equation*}
	v'_{x + t} = v_{x+t} + \beta 
		=
			\Bigl[
				f_k(x+t) + \sum\nolimits_\ell T_{k\ell}\int_{\BT_\ell} v_y\,dy - \beta
			\Bigr] + \beta
		> f_k(x) + \sum\nolimits_\ell T_{k\ell}\int_{\BT_\ell} v_y\,dy
		= v_x \geq v^\ast,
\end{equation*}
so that $\onebb_{v'_{x+t} < v^\ast} = 0$. The cross-share holdings term is unaffected if we replace $v$ by $v'$ as the integral is invariant with respect to changes in function value at isolated points,
\begin{equation*}
	\sum\nolimits_\ell T_{k\ell}\int_{\BT_\ell} v_y\,dy
		= \sum\nolimits_\ell T_{k\ell}\int_{\BT_\ell} v'_y\,dy,
\end{equation*}
and we conclude that
\begin{equation*}
	v'_{x+t} = f_k(x+t) + \sum\nolimits_\ell T_{k\ell}\int_{\BT_\ell} v'_y\,dy
		- \beta\onebb_{v'_{x+t} < v^\ast},
\end{equation*}
whence $v'$ is a new equilibrium. But then $v'_{x+t} > v_{x+t}$, leading to a contradiction of maximality.

Likewise, if $v$ is minimal, consider the effect of changing $v_x$ to $v''_x = v_x - \beta$ and keeping all other values the same, $v''_\xi = v_\xi$ if $\xi\neq x$. Arguing exactly as before, 
\begin{equation*}
	v''_x = v_x - \beta =
			\Bigl[
				f_k(x) + \sum\nolimits_\ell T_{k\ell}\int_{\BT_\ell} v_y\,dy
			\Bigr] - \beta
		< f_k(x+t) + \sum\nolimits_\ell T_{k\ell}\int_{\BT_\ell} v_y\,dy - \beta
		= v_{x+t} < v^\ast,
\end{equation*}
so that $\onebb_{v''_x < v^\ast} = 1$. We are thus led to the conclusion
\begin{equation*}
	v''_x = f_k(x) + \sum\nolimits_\ell T_{k\ell}\int_{\BT_\ell} v''_y\,dy
		- \beta\onebb_{v''_x < v^\ast},
\end{equation*}
whence we have yet another equilibrium. But now $v''_x < v_x$, leading to a contradiction of minimality.

Thus, in either case, $v\big|_{\BT_k}$ is increasing.

(ii) Introduce the nonce notation $\BT_k^\geq$ for the set consisting of the points $x$ in $\BT_k$ for which $v_x\geq v^\ast$ and let $\BT_k^< = \BT_k\setminus\BT_k^\geq$ consist of the remaining points in $\BT_k$ at which $v_x < v^\ast$. As $v\big|_{\BT_k}$ is monotone per part (i), $\BT_k^\geq$ and $\BT_k^<$ are disjoint intervals (with one of the two possibly empty or consisting of a singleton point). Writing $A_k = \sum_\ell T_{k\ell}\int_{\BT_\ell} v_y\,dy$ to compact notation, the fixed-point equation~\eqref{eq:univalue} satisfies $v_x = f_k(x) + A_k$ for $x$ in the sub-interval $\BT_k^\geq$, while it satisfies $v_x = f_k(x) + A_k - \beta$ for $x$ in the sub-interval $\BT_k^<$. The continuity of the endowment function $f_k$ on $\BT_k$ now carries the day and $v$ is continuous in each of the sub-intervals $\BT_k^\geq$ and $\BT_k^<$.

If either $\BT_k^\geq$ or $\BT_k^<$ is empty then $v$ is continuous everywhere on $\BT_k$. If both intervals are non-empty, by the monotonicity of $v$, there exists a unique point $x_k^\ast\in\BT_k$ such that $v_x\geq v^\ast$ for $x>x_k^\ast$ and $v_x< v^\ast$ for $x<x_k^\ast$. It follows that there is a jump discontinuity at $x_k^\ast$ with $v$ continuous on either side of it. To show that $v$ is continuous from the right it will suffice to show that $x_k^\ast\in\BT_k^{\geq}$. We will show a little bit more besides following the same pattern of construction as in part (i).


Suppose $v$ has a jump at the point $x_k^\ast$ in $\BT_k$. Then
\begin{equation*}
	v_x = \begin{cases} f_k(x) + A_k & \text{if $x_k^\ast < x \leq t_k$,}\\
			f_k(x) + A_k - \beta & \text{if $t_{k-1} < x < x_k^\ast$.} \end{cases}
\end{equation*}
Approaching $x_k^\ast$ from the right and the left in turn, by the continuity of $f_k$, we see that
\begin{equation}\label{eq:jump}
	v_{x_k^\ast+} = f_k(x_k^\ast) + A_k\;\;\text{and}\;\; 
	v_{x_k^\ast-} = f_k(x_k^\ast) + A_k - \beta.
\end{equation}
The first of the identities is vacuous if $x_k^\ast = t_k$ when the jump is on the boundary of $\BT_k$.

Suppose, to establish a contradiction, that $v_{x_k^\ast} < v^\ast$. Then $v_{x_k^\ast} = f_k(x_k^\ast) + A_k - \beta = v_{x_k^\ast-}$. But then $x_k^\ast$ cannot be on the boundary as else there is no jump. So $x_k^\ast$ must be an interior point of jump. Writing $\delta = v^\ast - v_{x_k^\ast}$ for the deficit, by the continuity of $f_k$, we may select $x_k^\ast < \eta\leq t_k$ so that $0 < f_k(\eta) - f_k(x_k^\ast) < \delta/2$. (Any factor $< 1$ could be chosen on the right; the choice of $1/2$ just keeps the algebra neat.) By construction, $\eta\in\BT_k^\geq$, whence $v_\eta = f_k(\eta) + A_k\geq v^\ast$.

If $v$ is maximal, mirror the construction in part (i) by forming a new function $v'$ by setting $v'_{x_k^\ast} = v_{x_k^\ast} + \beta$ and keeping all the other values of $v$ unchanged. Then
\begin{equation*}
    v'_{x_k^\ast} = v_{x_k^\ast} + \beta = f_k(x_k^\ast) + A_k = v_{x_k^\ast+}\geq v^\ast
\end{equation*}
and we have constructed an equilibrium larger than the maximal equilibrium.

If $v$ is minimal, form a new function $v''$ by setting $v''_\eta = v_\eta - \beta$ and keeping all other values of $v$ unchanged. Then 
\begin{multline*}
	v''_\eta = v_\eta - \beta = f_k(\eta) + A_k - \beta 
		= f_k(x_k^\ast) + A_k - \beta + \bigl(f_k(\eta) - f_k(x_k^\ast)\bigr)\\
		= v_{x_k^\ast} + \bigl(f_k(\eta) - f_k(x_k^\ast)\bigr)
		= v^\ast - (v^\ast - v_{x_k^\ast}) + \bigl(f_k(\eta) - f_k(x_k^\ast)\bigr)
		< v^\ast - \delta/2
\end{multline*}
and we have constructed an equilibrium smaller than the minimal equilibrium.

Having established a contradiction in all cases, we conclude that $v_{x_k^\ast}\geq v^\ast$ and, \emph{a fortiori}, $v\big|_{\BT_k}$ is continuous from the right.

(iii) It only remains to consider the behavior of the valuation at a point of jump. Suppose that $v$ has a jump at the point $x_k^\ast$ in $\BT_k$. We've established that $v_{x_k^\ast} \geq v^\ast$, whence $v_{x_k^\ast} = f(x_k^\ast) + A_k$. By~\eqref{eq:jump}, it follows that the size of the jump, $v_{x_k^\ast} - v_{x_k^\ast-}$, is exactly equal to the distress cost $\beta$. 

Now, on the one hand, $v_{x_k^\ast}\geq v^\ast$ from the conclusion of the proof of part (ii), while on the other, $v_{x_k^\ast}\leq v^\ast + \beta$ as a consequence of the just-established size $\beta$ of the jump. Indeed, if $v_{x_k^\ast} > v^\ast + \beta$ then $v_{x_k^\ast-} = v_{x_k^\ast} - \beta > v^\ast$ and $\BT_k^<$ is forced to contain points with valuations in excess of $v^\ast$ by the continuity of $v$ on the sub-interval $\BT_k^< = (t_{k-1}, x_k^\ast)$. But this contradicts the definition of $\BT_k^<$. We conclude that $v^\ast\leq v_{x_k^\ast}\leq v^\ast + \beta$.

Suppose $v$ is maximal and, to set up a contradiction, suppose further that $v_{x_k^\ast} > v^\ast$. Reusing notation, write $\delta = v_{x_k^\ast} - v^\ast$ for the excess. By the continuity of $f_k$ we may select a point $t_{k-1} < \xi < x_k^\ast$ in $\BT_k^<$ such that $0 < f_k(x_k^\ast) - f_k(\xi) < \delta/2$ and we so do. Now form a new function $v'$ by setting $v'_\xi = v_\xi + \beta$ and keeping all other values of $v$ unchanged. Then
\begin{multline*}
	v'_\xi = v_\xi + \beta = \bigl(f_k(\xi) + A_k - \beta\bigr) + \beta
		= f_k(x_k^\ast) + A_k - \bigl(f_k(x^\ast) - f_k(\xi)\bigr)\\
		= v_{x_k^\ast} - \bigl(f_k(x^\ast) - f_k(\xi)\bigr)
		= v^\ast + (v_{x_k^\ast} - v^\ast) - \bigl(f_k(x^\ast) - f_k(\xi)\bigr)
		> v^\ast + \delta/2
\end{multline*}
and we have constructed an equilibrium larger than the maximal equilibrium. With our hypothesis leading to a contradiction, it must hence be the case that $v_{x_k^\ast} = v^\ast$ if $v$ is maximal.

The case when $v$ is minimal is even simpler. To set up a contradiction again, suppose that $v_{x_k^\ast} < v^\ast + \beta$. Form a new function $v''$ by setting $v''_{x_k^\ast} = v_{x_k^\ast} - \beta$ and keeping all other values of $v$ unchanged. Then $v''_{x_k^\ast} = f_k(x_k^\ast) + A_k - \beta < v^\ast\leq v_{x_k^\ast}$ and we have constructed an equilibrium smaller than the minimal equilibrium. As this is impossible, we conclude that $v_{x_k^\ast} = v^\ast + \beta$ if $v$ is minimal.
This concludes the proof.
\end{proof}

\begin{proof}[Proof of Theorem~\ref{thm:maximaleq}]
We assume that $v_x$ is a maximal equilibrium. The minimal equilibrium case proceeds by a straightforward modification of the same arguments.

Let $\boldsymbol{\kappa}$ be a putative solvency vector and recall that $\mathbf{V}(\boldsymbol{\kappa})$ are the corresponding values, so that $\mathbf{V}(\boldsymbol{\kappa})$ solves $$     	\bV(\bkappa) = \be + \bC\bV(\bkappa) - \beta(\mathbf{1} - \bkappa).\\
$$
Given a vector of firms $\mathbf{y}^*=(y^*_k)_{k=1}^m$ with one in each block, we define a putative solvency vector $\boldsymbol{\kappa}_{(y_k^*)}$ for a finite random network  by $({\kappa}_{(y_k^*)})_i = \onebb_{i \geq y_k^*}$ for all $i$ in block $k$. That is, a firm in block $k$ is solvent if and only if its index is at least the threshold $y_k^*$.

Similarly, given a vector of firms $\boldsymbol{y}^* = (y^*_k)_{k=1}^m$, recall we define a putative equilibrium for the graphon  $\kappa(x; \boldsymbol{y}^*)$ by supposing firms $x \in \BT_k$ are solvent if and only if $x \geq y_k^*$. Recall we refer to the corresponding values as
$v\bigl(\kappa(\,\cdot\,; \boldsymbol{y}^*)\bigr)$ and that these values solve
\begin{equation*}
	v\bigl(\kappa(\,\cdot\,;\boldsymbol{y}^*)\bigr)_x = f_k(x) + \sum_{\ell=1}^m T_{k\ell} \int_{\BT_\ell} v\bigl(\kappa(\,\cdot\,; \boldsymbol{y}^*)\bigr)_y\,dy 
		- \beta(1- \kappa(x; \boldsymbol{y}^*)).
\end{equation*}

Given $\boldsymbol{\kappa}$, let $\overline{V}(\boldsymbol{\kappa})$ be the values of firms on a finite network with deterministic link weights $g_{k\ell}$ and firm failures described by the putative solvency vector $\boldsymbol{\kappa}$. We begin with a lemma establishing several properties of $v(\boldsymbol{\kappa})$, $V(\boldsymbol{\kappa})$, and $\overline{V}(\boldsymbol{\kappa}).$

\begin{lemma}\label{lem:value-properties}
    (i) $v(\boldsymbol{\kappa})$, $V(\boldsymbol{\kappa})$, and $\overline{V}(\boldsymbol{\kappa})$ are increasing in $\kappa$,
    
    (ii) $v(\boldsymbol{\kappa}_{(y_k^*)})_x$ is continuous in $x$ and $y_k^*$ in any neighborhood where $(\kappa_{(y_k^*)})_x$ is constant and $x$ is contained in a single block,
    
    (iii) $\overline{V}(\boldsymbol{\kappa}_{(y^*_k)})_i\rightarrow v\bigl(\kappa(\,\cdot\,; \boldsymbol{y}^*)\bigr)_i$ uniformly in $i$.
\end{lemma}

\begin{proof}
(i) We prove monotonicity by expressing each of the values in terms of firm endowments and failure costs. Given $x \in \BT_r$, we have
\begin{align*}
v(\boldsymbol{\kappa})_x & = e_x + \sum_{k} T_{rk} \int_{\BT_k} v_y dy - \beta (1-\kappa_x)
\\ & = e_x + A_r - \beta (1-\kappa_x)
\end{align*}
where $A_r =  \sum_{k} T_{rk} \int_{\BT_k} v_y dy $.
Substituting for $v_y$, we obtain
\begin{align*}A_r & = \sum_k T_{rk} \int_{\BT_k} (e_y - \beta(1-\kappa_y) + A_k)dy
\\ & = \sum_k T_{rk} (s_k \overline{e}_k+ s_k A_k - \beta \int_{\BT_k}(1-\kappa_y)dy).\end{align*}
Solving and expressing the $A_r$ in vector form,
\begin{equation}\label{eq:valuesgraphonlinear}\begin{pmatrix}A_1 \\ \vdots \\ A_m \end{pmatrix} = (I-\bT D)^{-1}\bT\begin{pmatrix} s_1 \overline{e}_1 -  \beta \int_{\BT_1}(1-\kappa_y)dy \\ \vdots \\s_m \overline{e}_m - \beta \int_{\BT_m}(1-\kappa_y)dy \end{pmatrix}.\end{equation}
We conclude that $A_r$ is a linear combination of the $\overline{e}_k$ and $-\int_{\BT_k} (1-\kappa_y)dy$ with non-negative coefficients. Since $-\int_{\BT_k} (1-\kappa_y)dy$ is increasing in $\boldsymbol{\kappa}$, so are the $A_r$ and therefore the values $v(\boldsymbol{\kappa})_x.$

Similarly, we have
$$\boldsymbol{V}(\boldsymbol{\kappa}) = \boldsymbol{e} + \boldsymbol{C} \boldsymbol{V}(\boldsymbol{\kappa})- \beta (\boldsymbol{1}-\boldsymbol{\kappa})$$
and therefore
$$\boldsymbol{V}(\boldsymbol{\kappa}) =(\boldsymbol{I}-\boldsymbol{C})^{-1}( \boldsymbol{e} - \beta (\boldsymbol{1}-\boldsymbol{\kappa})).$$
We conclude that each entry of $\boldsymbol{V}(\boldsymbol{\kappa})$ is a linear combination of the entries of $\boldsymbol{e}$ and $-(\boldsymbol{1}-\boldsymbol{\kappa})$ with non-negative coefficients. Since $-(\boldsymbol{1}-\boldsymbol{\kappa})$ is increasing  in $\boldsymbol{\kappa}$, so is $\boldsymbol{V}(\boldsymbol{\kappa})$.

Finally, we have
$$\overline{\boldsymbol{V}}(\boldsymbol{\kappa}) = \boldsymbol{e} + \overline{\boldsymbol{C}} \,\overline{\boldsymbol{V}}(\boldsymbol{\kappa}) - \beta (\boldsymbol{1}-\boldsymbol{\kappa})$$
and therefore
$$\overline{\boldsymbol{V}}(\boldsymbol{\kappa}) =(\boldsymbol{I}-\overline{\boldsymbol{C}})^{-1}( \boldsymbol{e} - \beta (\boldsymbol{1}-\boldsymbol{\kappa})).$$
We conclude that each entry of $\overline{\boldsymbol{V}}(\boldsymbol{\kappa})$ is a linear combination of the entries of $\boldsymbol{e}$ and $-(\boldsymbol{1}-\boldsymbol{\kappa})$ with non-negative coefficients. Since $-(\boldsymbol{1}-\boldsymbol{\kappa})$  is increasing in $\boldsymbol{\kappa}$, so is $\overline{\boldsymbol{V}}(\boldsymbol{\kappa})$.

(ii) This follows from equation \eqref{eq:valuesgraphonlinear} since each term $\int_{\BT_k} (1-\kappa_y)dy$ is continuous in the thresholds $y_k^*$ and the endowment $e_x$ is continuous in $x$.

(iii) Let $\boldsymbol{\kappa} = \kappa(\cdot;(y_k^*))$ on the graphon and $\boldsymbol{\kappa} = \boldsymbol{\kappa}_{(y_k^*)}$ in the finite case. From equation \eqref{eq:valuesgraphonlinear}, we have
\begin{equation}\label{eq:valuesseriesexpansion}v(\boldsymbol{\kappa})_i  = e_i - \beta (1-\kappa_i) + 
\left[(I-\bT D)^{-1}\bT\begin{pmatrix} s_1 \overline{e}_1 -  \beta \int_{\BT_1}(1-\kappa_y)dy \\ \vdots \\s_m \overline{e}_m - \beta \int_{\BT_m}(1-\kappa_y)dy\end{pmatrix}\right]_r.\end{equation}
We also have
$$\overline{V}(\boldsymbol{\kappa})_i = e_i - \beta (1-\kappa_i) + [(I-\overline{\boldsymbol{C}})^{-1}\overline{\boldsymbol{C}}( \boldsymbol{e} - \beta (\boldsymbol{1}-\boldsymbol{\kappa}))]_i.$$
The entries of $n(I-\overline{\boldsymbol{C}})^{-1}\overline{\boldsymbol{C}}$ converge uniformly to the corresponding entries of $(I-\bT D)^{-1}\bT$, as the differences are due to the addition or deletion of a bounded number of firms near the boundaries of the intervals $\BT_k$. Since the Riemann sums $$\sum_{k} \sum_{i \in \BT_k} [(I-D \bT)^{-1}\bT]_{rk} \cdot \frac{e_i -\beta (1-\kappa_i)}{n}$$
converge to the integral $$\sum_{k} \int_{y \in \BT_k} [(I-D \bT)^{-1}\bT]_{rk} (e_y -\beta (1-\kappa_y))dy,$$
the desired limit holds. The convergence is uniform in $i$ since there are only finitely many blocks and therefore finitely many such sequences of Riemann sums.
\end{proof}

The proof of Theorem \ref{thm:fixedkappa} shows that there exists a sequence $h_1(n)\rightarrow 0$ such that the probability that $$\max_i |\overline{V}(\boldsymbol{\kappa}_{(y_k^*)})_i -V(\boldsymbol{\kappa}_{(y^*_k)})_i| > h_1(n)$$
decays at an exponential rate in $n$ (with constants independent of the firms $(y^*_k)$). 
We also have the deterministic limit $$\max_i |v\bigl(\kappa(\,\cdot\,; \boldsymbol{y}^*)\bigr)_i -\overline{V}(\boldsymbol{\kappa}_{(y^*_k)})_i| \rightarrow 0$$
by Lemma \ref{lem:value-properties}(iii).

Combining these limits and applying the triangle inequality, we can choose a sequence $h_2(n) \rightarrow 0$ such that probability that
$$\max_i |v\bigl(\kappa(\,\cdot\,; \boldsymbol{y}^*)\bigr)_i -V(\boldsymbol{\kappa}_{(y^*_k)})_i| > h_2(n)$$
also decays at an exponential rate in $n$.

Now let $\mathcal{Y}_n$ be the set of vectors of firms $\mathbf{y}^*=(y^*_k)_{k=1}^m$ such that each $y^*_k$ can be expressed as a rational number with denominator at most $n$. The cardinality of $\mathcal{Y}_n$ is polynomial in $n$, so 
$$\max_{\mathbf{y}^* \in \mathcal{Y}_n} \max_i |v\bigl(\kappa(\,\cdot\,; \boldsymbol{y}^*)\bigr)_i -V(\boldsymbol{\kappa}_{(y^*_k)})_i|\rightarrow 0$$
almost surely.

For the remainder of the proof, fix a sequence of realized random networks for which this random variable converges to zero. We will show that for all $k$, along this sequence $\overline{i}_k(n) \rightarrow x_k^*$ and  $\underline{i}_k(n) \rightarrow x_k^*$. This will imply that $\overline{i}_k(n) \rightarrow x_k^*$ almost surely and  $\underline{i}_k(n) \rightarrow x_k^*$ almost surely for all $k$.

\textbf{Step 1}: $\limsup \overline{i}_k(n) \leq x_k^*$ for all $k$.

Suppose for the sake of contradiction that $\limsup_n \overline{i}_{k'}(n) > x_{k'}^*$ for some $k'$. Passing to a convergent subsequence, we can assume that $\overline{i}_{k'}(n)$ converges to some $\overline{i}_{k'} > x_{k'}^*$.

Let $B$ be the spillover matrix. We next prove a matrix algebra lemma about $B$.

\begin{lemma}
Suppose $B$ is a non-negative $n \times n$ matrix with $\|B\|_2<1$. There exists a vector $\mathbf{w}$ with all entries positive and real and a constant $\gamma<1$ such that $(B w)_k \leq \gamma w_k$ for all $k$.
\end{lemma}

\begin{proof}
Because $\|\cdot\|_2$ is continuous, we can choose $B'$ such that $B'_{ij} > 0$ for all $i$ and $j$, $B'_{ij}\geq B_{ij}$ for all $i$ and $j$, and $\|B'\|_2 < 1$.

By the Perron-Frobenius theorem, we can choose a vector $w$ with all entries real and positive such that $B'w=\|B'\|_2w$. Taking $\gamma=\|B'\|_2,$ we have
$$(B w)_k \leq (B'w)_k = \gamma w_k $$
as desired.
\end{proof}

Choose $w$ and $\gamma<1$ as in the preceding lemma and  consider the thresholds $\boldsymbol{x}^*+tw$ for $t>0$. Define $ \lambda\bigl(\BI_k(t; \bx^\ast + tw)\bigr)$ to be the Lebesgue measure of firms in block $k$ above the threshold $x_k^*$ with values under the putative equilibrium $\boldsymbol{\kappa}(\boldsymbol{x}^*+tw)$ below $v^*$. Because the derivative of $ \lambda\bigl(\BI_k(t; \bx^\ast + tw)\bigr)$ at $t=0$ is linear in each $1/f_k'(x_k^*)$, we have	\begin{align*}	\frac{d}{dt} \lambda\bigl(\BI_k(t; \bx^\ast + tw)\bigr)\bigg|_{t=0} & = \sum_{\ell} w_{\ell} \frac{d}{dt} \lambda\bigl(\BI_k(t; \bx^\ast + t \onebb_\ell\bigr)\bigg|_{t=0}
\\ & = \sum_{\ell} w_{\ell} B_{ k\ell}
\\ & = (B w)_k 
\\ & \leq \gamma w_k\end{align*}
where the last line follows from our choice of $w$ and $\gamma$. Because the $f_k'(x)$ are continuously differentiable, we can choose $\alpha$ with $1<\alpha<1/\gamma$ and $\overline{t}>0$ such that $$\lambda\bigl(\BI_k(t; \bx^\ast + tw)\bigr)< \alpha \gamma w_k$$
for all $k$ and all $t < \overline{t}$. There exists $\epsilon>0$ such that $$v\bigl(\kappa(\,\cdot\,; \boldsymbol{x}^*+tw)\bigr)_{x_k^*+tw_k}>v^*+\epsilon$$
for all $k$ and all $t < \overline{t}$. By monotonicity (Lemma \ref{lem:value-properties}(i)),$$v\bigl(\kappa(\,\cdot\,; \boldsymbol{x}^*+tw)\bigr)_{x}>v^*+\epsilon$$
for all $x \in \BT_k$ with $x \geq x_k^*+tw_k$ and all $t < \overline{t}$.

Recall we chose $k'$ such that $\overline{i}_{k'}>x_{k'}^*$. So we can fix $t < \overline{t}$ such that $x^*_{k'}+tw_{k'} < \overline{i}_{k'}$. By Lemma \ref{lem:value-properties}(ii) and the density of the rational numbers, we can choose $\mathbf{y}^* = (y_k^*)_{k=1}^m$ with all $y_k^*$ rational and $y_k^*<x_k^*+tw_k$ such that $$v\bigl(\kappa(\,\cdot\,; \boldsymbol{y}^*)\bigr)_{x}>v^*+\epsilon$$
for all $x \in \BT_k$ with $x \geq y_k^*$. By our choice of a sequence of realized random networks above, we have
$$\max_i |v\bigl(\kappa(\,\cdot\,;y_k^*)\bigr)_i -V(\boldsymbol{\kappa}_{(y_k^*)})_i| \rightarrow 0.$$

So for $n$ sufficiently large, $$\max_i |v\bigl(\kappa(\,\cdot\,; \boldsymbol{y}^*)\bigr)_i -V(\boldsymbol{\kappa}_{(y_k^*)})_i| <\epsilon.$$ By the triangle inequality, for $n$ sufficiently large $V(\boldsymbol{\kappa}_{(y_k^*)})_i \geq v^*$ whenever $i \in \BT_k$ with $i \geq y_k^*$. We will show that this implies that the same holds at the maximal equilibrium.

The profile $\boldsymbol{\kappa}_{(y_k^*)}$ need not be an equilibrium, but we have shown that $V(\boldsymbol{\kappa}_{(y_k^*)})_i\geq v^*$ whenever $(\boldsymbol{\kappa}_{(y_k^*)})_i = 1$. Therefore, we can obtain an equilibrium from $\boldsymbol{\kappa}_{(y_k^*)}$ as follows. Set $\boldsymbol{\kappa}^{(0)} = \boldsymbol{\kappa}_{(y_k^*)}$. Then, repeatedly define $\boldsymbol{\kappa}^{(\ell+1)}$ by beginning with $\boldsymbol{\kappa}^{(\ell)}$ and setting entries $(\kappa^{(\ell+1)})_i$ to $1$ if $\mathbf{V}(\boldsymbol{\kappa}^{(\ell)})_i \geq v^*$.  Because $\mathbf{V}(\boldsymbol{\kappa})$ is monotone increasing in $\boldsymbol{\kappa}$ (Lemma \ref{lem:value-properties}(i)), this process does not decrease any firms' values at any point. So for all $\ell$ we have $\mathbf{V}(\boldsymbol{\kappa}^{(\ell)})_i \geq v^*$ whenever $\kappa^{(\ell)}_i = 1$.

This process converges in at most $n$ steps. By construction the limit $\boldsymbol{\kappa}^{(n)}$ must satisfy $\kappa^{(n)}_i = 1$ whenever $\mathbf{V}(\boldsymbol{\kappa}^{(n)})_i \geq v^*$. Since we noted in the previous paragraph that $\mathbf{V}(\boldsymbol{\kappa}^{(n)})_i \geq v^*$ whenever $\kappa^{(n)}_i = 1$, the putative solvency vector $\boldsymbol{\kappa}^{(n)}$ is feasible and therefore is indeed an equilibrium.

For $n$ sufficiently large $V(\boldsymbol{\kappa}^{(n)})_i \geq v^*$ whenever $i \in \BT_k$ with $i \geq y_k^*$. Since  $y_k^*<x^*_{k'}+tw_{k'} < \overline{i}_{k'}$, for $n$ sufficiently large we must have $\kappa^{(n)}_i=1$ for all $i$ in an open neighborhood of $\overline{i}_{k'}$. But any firm that is solvent under the equilibrium $\boldsymbol{\kappa}^{(n)}$ is also solvent under the maximal equilibrium. So this contradicts the definition of $\overline{i}_{k'}$. We conclude that $\limsup_n \overline{i}_k(n) \leq x_k^*$ for all $k$, completing step 1.

\textbf{Step 2}: $\liminf_n \underline{i}_k(n) \geq x_k^*$ almost surely for all $k$.

Suppose for the sake of contradiction that $\liminf_n \underline{i}_{k'}(n) < x^*_{k'}$ for some $k'$. Passing to a subsequence, we can assume that $\underline{i}_{k}(n)$ converges to some $\underline{i}_k$ for each $k$ and that $\underline{i}_{k'} < x_{k'}^*$.

We write $\boldsymbol{\kappa}^{max}$ for the solvency vector at the maximal equilibrium for the realized random network with $n$ agents. By the definition of $\underline{i}_k(n)$, each firm $\underline{i}_k(n)$ has value $V(\boldsymbol{\kappa}^{max})_{\underline{i}_k(n)} \geq v^*$. The putative solvency vector $\boldsymbol{\kappa}_{(\underline{i}_k(n))} \geq \boldsymbol{\kappa}^{max}$, so by monotonicity (Lemma \ref{lem:value-properties}(i)) we have $V(\boldsymbol{\kappa}_{(\underline{i}_k(n))})_{\underline{i}_k(n)} \geq v^*$.

Note that the values $\underline{i}_k(n)$ are all rational with denominators at most $n-1$, so $(\underline{i}_k(n))_{k=1}^m \in \mathcal{Y}_n$.  By our choice of a sequence of realized random networks above, we have
$$\max_i |v\bigl(\kappa(\,\cdot\,; (\underline{i}_k(n)))\bigr)_i -V(\boldsymbol{\kappa}_{(y^*_k)})_i|\rightarrow 0,$$
and therefore
$$\liminf_n v\bigl(\kappa(\,\cdot\,; (\underline{i}_k(n)))\bigr)_{\underline{i}_k(n)} \geq v^*$$
for all $k$. Since $\underline{i}_k(n) \rightarrow \underline{i}_k$ for all $k$, by Lemma \ref{lem:value-properties}(ii) we have
$v\bigl(\kappa(\,\cdot\,; (\underline{i}_k)\bigr)_{\underline{i}_k} \geq v^*$ for all $k$. Applying Lemma \ref{lem:value-properties}(i) again, we have $v\bigl(\kappa(\,\cdot\,; (\underline{i}_k)\bigr)_x \geq v^*$ whenever $x \in \BT_k$ satisfies $x \geq \underline{i}_k$.

We claim that this implies that the maximal equilibrium value $v_{\underline{i}_k} \geq v^*$ for all $k$. Since we assumed at the start of Step 2 that there exists a $k'$ such that $\underline{i}_{k'} < x_{k'}^*,$ this will give a contradiction.

The claim is a special case of the following lemma:
\begin{lemma}\label{lem:eq_construction}
Let ${\kappa}$ be a putative solvency function. If $v(\kappa)_x \geq v^*$ whenever  $\kappa_x = 1$, then each firm that is solvent under $\kappa$ is solvent under the maximal equilibrium.
\end{lemma}

\begin{proof} To prove the lemma, recall the order-preserving operator $\tau$ from the proof of Proposition~\ref{prop:Tarski}. Let $\mathcal{F}^*$ be the set of measurable, bounded functions $w\colon (0,1]\to\real$ with $w \geq v(\kappa)$. Because $v(\kappa)_x \geq v^*$ whenever $x\in \BT_k$ satisfies $\kappa_x=1$, for any $w \in \mathcal{F}^*$ we have $w_x \geq v^*$ whenever $x \in \BT_k$ satisfies $\kappa_x=1$. 

Suppose that $w \in \mathcal{F}^*$. Then for $x \in \BT_k$, \begin{align*}(\tau w)_x & =  e_x + \int_0^1 C(x,y) w_y\,dy - \beta\onebb_{w_x<v^*}
\\ & \geq e_x +  \int_0^1 C(x,y) w_y\,dy - \beta(1-\kappa_x)
\\ & \geq e_x + \int_0^1 C(x,y) v(\kappa)_y\,dy - \beta (1-\kappa_x) \text{ by the definition of }\mathcal{F}^*
\\ & = v(\kappa)_x.
\end{align*}
This shows that $\tau w \in \mathcal{F}^*$, so $\tau$ maps the set $\mathcal{F}^*$ to itself. Applying Tarksi's fixed point theorem as in the proof of Proposition~\ref{prop:Tarski}, we conclude there is an equilibrium $w_x$ in $\mathcal{F}^*$. The  equilibrium values $w_x$ satisfy $w_x \geq v(\kappa)_x \geq v^*$ whenever $\kappa_x=1$. So the maximal equilibrium also satisfies $v_{x} \geq v^*$ whenever $\kappa_x=1$.
\end{proof}
We have established a contradiction, so we can conclude that $\liminf_n \underline{i}_k(n) \geq x_k^*$ for all $k$.

\textbf{Step 3}: $\lim \underline{i}_k(n) =\lim \overline{i}_k(n) = x_k^*$ for all $k$.

Fix $k$. We have shown that $\limsup \overline{i}_k(n) \leq x_k^*$ and $\liminf_n \underline{i}_k(n) \geq x_k^*$. So we can choose $\epsilon(n)\rightarrow 0$ such that all firms in $i \in \BT_k$ that are insolvent have $i< x_k^*+\epsilon(n)$ and all firms $i \in \BT_k$ that are solvent have $i> x_k^*-\epsilon(n)$. Equivalently, all firms with $i>x_k^*+\epsilon(n)$ are solvent and all firms with $i<x_k^*-\epsilon(n)$ are insolvent. So $\underline{i}_k(n), \overline{i}_k(n) \in [x_k^*-\epsilon(n),x_k^*+\epsilon(n)]$ for all $n$. Since $\epsilon(n) \rightarrow 0$, this implies that $\lim \underline{i}_k(n) =\lim \overline{i}_k(n) = x_k^*$.
\end{proof}

\begin{proof}[Proof of Proposition~\ref{prop:bailout}] We begin with a lemma, which states that any firm receiving cash under an optimal infusion must have value $v^*$ after the infusion.
\begin{lemma}\label{lem:bailout}
If $\iota(x)$ is an optimal cash infusion, then, $\iota(x)>0 \Rightarrow \,\, \widetilde{v}_x = v^*$.
\end{lemma}
\begin{proof}
Suppose $\iota(x)>0$ for a positive measure of firms with $\widetilde{v}_x > v^*$. Choose a positive measure subset $S$ of these firms contained in some block, $k$, say. Consider any cash infusion $\iota_0(x)$ obtained by decreasing $\iota(x)$ by $\widetilde{v}_x - v^*$ for all firms in $S$ and increasing $\iota(x)$ for other firms in block $k$ by a total of $\int_S (\widetilde{v}_x - v^*)dx$.

Let $\kappa$ be the solvency function corresponding to the maximal equilibrium after cash infusion $\iota(x)$. We write $w(\kappa)$ for the values under putative solvency function $\kappa$ after cash infusion $\iota_0(x)$. Applying equation~\eqref{eq:valuesseriesexpansion} with endowments $e_x+\iota_0(x)$, we find that for any $x \in \BT_k$
\begin{align*}
w(\kappa)_x & =  e_x + \iota_0(x) - \beta (1-\kappa_x) + 
\left[(I-\bT D)^{-1}\bT\begin{pmatrix} s_1 \overline{e}_x + \int_{\BT_1} \iota_0(y)dy -  \beta \int_{\BT_1}(1-\kappa_y)dy \\ \vdots \\s_m \overline{e}_m + \int_{\BT_m} \iota_0(y)dy  - \beta \int_{\BT_m}(1-\kappa_y)dy\end{pmatrix}\right]_k
\\ & =  e_x + \iota_0(x) - \beta (1-\kappa_x) + 
\left[(I-\bT D)^{-1}\bT\begin{pmatrix} s_1 \overline{e}_x + \int_{\BT_1} \iota(y)dy -  \beta \int_{\BT_1}(1-\kappa_y)dy \\ \vdots \\s_m \overline{e}_m + \int_{\BT_m} \iota(y)dy  - \beta \int_{\BT_m}(1-\kappa_y)dy\end{pmatrix}\right]_k
\\ & = \widetilde{v}_x+\iota_0(x)-\iota(x).
\end{align*}
The second equality holds because  $ \int_{\BT_k} \iota_0(y)dy=\int_{\BT_k} \iota(y)dy$ for all $k$. By the construction of $\iota_0(x)$, we have
$$\widetilde{v}_x+\iota_0(x)-\iota(x) \geq v^*$$
whenever $\widetilde{v}_x \geq v^*$. So $w(\kappa)_x \geq v^*$ whenever $\widetilde{v}_x \geq v^*$. By Lemma~\ref{lem:eq_construction}, this implies that any firm that is solvent at the maximal equilibrium after cash infusion $\iota(x)$ is solvent at the maximal equilibrium after cash infusion $\iota_0(x)$. We have shown that $\iota_0(x)$ weakly increases the measure of firms solvent at the maximal equilibrium relative to $\iota(x)$.

It remains to obtain a strict improvement. By assumption, the support of $\iota(x)$ is interior in block $k$. The construction of $\iota_0(x)$ thus far did not depend on how we reallocated the amount $\int_S (\widetilde{v}_x - v^*)dx$ within block $k$. We can do so to obtain a positive measure of firms such that $$\widetilde{v}_x+\iota_0(x)-\iota(x) \geq v^*$$
but $\widetilde{v}_x < v^*$. Then, Lemma~\ref{lem:eq_construction} similarly shows these firms are solvent at the maximal equilibrium under $\iota_0(x)$. So, the set of firms solvent after $\iota(x)$ is a strict subset of the set of firms solvent after $\iota_0(x)$. This contradicts the optimality of the original cash infusion $\iota(x)$.

So we cannot have $\iota(x)>0$ for a positive measure of firms with $\widetilde{v}_x > v^*$. Essentially the same argument shows that we cannot have a positive measure of firms with $\iota(x)>0$ and $\widetilde{v}_x < v^*$, as we can reduce the cash infusion for those firms from $\iota(x)$ to $0$ and increase the cash infusion to other firms in the same block. So we must have $\widetilde{v}_x=v^*$ whenever $\iota(x)>0$.
\end{proof}

Within each block $k$, the value of each firm $x$ with $\iota(x)>0$ after the cash infusion is equal to $e_x + \iota(x) + \xi_k$ for some constant $\xi_k$ depending on the block $k$. Lemma \ref{lem:bailout} implies that $e_x+\iota(x)$ is equal to some value, which we call $\widehat{e}_k$, for all firms in block $k$ with $\iota(x)>0$.

Within each block, the set of firms that receive an injection form an interval. If not, there is a block $k$ where the set $x$ of firms for which $\iota(x)>0$, is not an interval. We show that there is another cash transfer with the same budget under which more firms are solvent.

To see this, let $\underline{\iota}_k$ and $\overline{\iota}_k$ be the infimum and supremum, respectively, of the support of $\iota(x)$ in the block $\BT_k$. Since the set where $\iota(x)>0$ is not an interval (up to measure zero sets), we can choose a positive measure subset of $[\underline{\iota}_k, \overline{\iota}_k]$ where $\iota(x)=0$. We can also choose $\epsilon>0$ such that $$S = \{x \in [\underline{\iota}_k ,\underline{\iota}_k + \epsilon]: \iota(x)>0\}$$
and $$S' = \{x \in [\underline{\iota}_k + \epsilon, \overline{\iota}_k]: \iota(x)=0\}$$
both have positive measure.

By construction we have $x < x'$ for all $x \in S$ and $x' \in S'$, $\iota(x)>0$ for all $x \in S$, and $\iota(x)=0$ but $\widetilde{v}_x<v^*$ for all $x \in S'$. Shrinking $S$ and $S'$ if necessary, we can assume that $$ \int_S (\widehat{e}_k- e_x)dx=\int_{S'}(\widehat{e}_k- e_x)dx.$$
Since $x < x'$ and so $e_x < e_{x'}$ for all $x \in S$ and $x' \in S'$, this implies the Lebesgue measures of these sets satisfy $\lambda(S) < \lambda(S')$. Then define a cash infusion $\iota_0(x)$ with the same budget as $\iota(x)$ by setting $\iota_0(x)=0$ for all $x\in S$, $\iota_0(x) = \widehat{e}_k-e_x$ for all $x \in S'$, and $\iota_0(x) = \iota(x)$ for all $x$ not in $S$ or $S'$. So $\iota_0$ differs from $\iota$ by providing cash to firms in $S'$ instead of firms in $S$.

Let $\kappa$ be the solvency function associated with the maximal equilibrium after cash infusion $\iota(x)$. Define $\kappa_0$ by $(\kappa_0)_x=\kappa_x$ for $x$ not in $S$ or $S'$, $(\kappa_0)_x=0$ for $x \in S$, and $(\kappa_0)_{x'}=1$ for $x'\in S'$. Write $w(\kappa_0)$ for the values under putative solvency function $\kappa_0$, after cash infusion $\iota_0(x)$. Observe that strictly more firms are solvent under $\kappa_0$ than under $\kappa$ since $\lambda(S')>\lambda(S)$.

We proceed similarly to the proof of Lemma~\ref{lem:bailout}. Applying equation~\eqref{eq:valuesseriesexpansion} with endowments $e_x+\iota_0(x)$, we find that for any $x \in \BT_k$
\begin{align*}
w(\kappa)_x & =  e_x + \iota_0(x) - \beta (1-(\kappa_0)_x) + 
\left[(I-\bT D)^{-1}\bT\begin{pmatrix} s_1 \overline{e}_x + \int_{\BT_1} \iota_0(y)dy -  \beta \int_{\BT_1}(1-(\kappa_0)_y)dy \\ \vdots \\s_m \overline{e}_m + \int_{\BT_m} \iota_0(y)dy  - \beta \int_{\BT_m}(1-(\kappa_0)_y)dy\end{pmatrix}\right]_k
\\ & \geq  e_x + \iota_0(x) - \beta (1-\kappa_x) + 
\left[(I-\bT D)^{-1}\bT\begin{pmatrix} s_1 \overline{e}_x + \int_{\BT_1} \iota(y)dy -  \beta \int_{\BT_1}(1-\kappa_y)dy \\ \vdots \\s_m \overline{e}_m + \int_{\BT_m} \iota(y)dy  - \beta \int_{\BT_m}(1-\kappa_y)dy\end{pmatrix}\right]_k
\\ & = \widetilde{v}_x+\iota_0(x)-\iota(x).
\end{align*}
The inequality holds because $\int_{\BT_k}(1-(\kappa_0)_y)dy \leq \int_{\BT_k}(1-\kappa_y)dy$ for all $k$ by the choice of $\kappa_0$. The final expression is at least $v^*$ whenever $(\kappa_0)_x=1$. By Lemma~\ref{lem:eq_construction}, this implies that any firm with $(\kappa_0)_x=1$ is solvent at the maximal equilibrium after cash infusion $\iota_0(x)$. But since strictly more firms are solvent under $\kappa_0$ than under $\kappa$, this contradicts the optimality of $\iota(x)$. We conclude the set of firms in block $k$ with $\iota(x)>0$ is an interval.
\end{proof}
\begin{proof}[Proof of Theorem~\ref{t:bailout_linear}]
We begin by deriving the relationship between $y_k$ and $y_{k'}$. Suppose that we marginally change the total amount of the cash infusion to block $k$ while adjusting the intervals of firms with $\iota(x)>0$ in each block so that Lemma~\ref{lem:bailout} continues to hold. The optimality of the original cash infusion $\iota(x)$ implies that the marginal benefit of this change is independent of the block $k$.

Suppose we increase the cash infusion to block $k$ by $t$. The derivative at $t=0$ of the measure of additional firms directly made solvent is equal to the reciprocal of the change in endowment at the left endpoint of the interval of firms with $\iota(x)>0$. Because this interval has length $y_k$ and endowments have slope $a_k$, this change in endowment is $a_ky_k$. So its reciprocal is $1/(a_ky_k)$. The derivative of the corresponding change in the value of cross-holdings in block $k$ is equal to the change in this value from a measure $1/(a_ky_k) + 1/\beta$ of additional firms failing, where the second term corresponds to the cash value of the additional cash infusion.

So (as in the proof of Theorem~\ref{t:spillovermatrix}), the derivatives at $t=0$ of the measures of firms in each interval that become solvent as a result of the corresponding change in value are given by 
$$ B E_k (1/(a_ky_k) + 1/\beta).$$
As a consequence of these additional solvent firms,
$$ B^2 E_k (1/(a_ky_k) + 1/\beta)$$
further firms become solvent. Iterating, the derivatives at $t=0$ of the total measure of firms that become solvent as an indirect result of the additional cash infusion is
\begin{equation}\label{eq:marginal_indirect}\sum_{\ell=1}^{\infty} \mathbf{1}^{\top} B^{\ell} E_k (1/(a_ky_k) + 1/\beta) = \mathbf{1}^{\top} (I-B)^{-1} B E_k (1/(a_ky_k) + 1/\beta).\end{equation}

Optimality of $\iota(x)$ requires that this measure is equal across blocks $k$, as otherwise there would be a profitable deviation reducing the cash infusion in one block and increasing it in another. This implies the formula in the proposition.

We next derive the budget constraint. For each block $k$, we have $\iota(x) = a_k (x_k^*-\delta_k-x)$ for $x \in [x_k^*-\delta_k-y_k,x_k^*-\delta_k]$ and $\iota(x)=0$ otherwise. So the integral of $\iota(x)$ in each block $k$ is $\frac12a_ky_k^2$, and therefore the total cost of the infusion is $\frac12\sum_k a_ky_k^2$. Since the optimal intervention is interior, the budget constraint binds and therefore this total cost must be $K$.

The expression for $\delta_k$ remains. By the same logic as in equation (\ref{eq:marginal_indirect}), the total measure of firms in each block that become solvent as an indirect result of the cash infusions are given by the vector $$ (I-{B})^{-1}B \begin{pmatrix}
    y_1 + \frac{a_1y_1^2}{2\beta}
    \\  \vdots \\ y_m + \frac{a_my_m^2}{2\beta}\end{pmatrix}.$$Taking the $k^{th}$ entry of this vector gives $\delta_k$.
\end{proof}

\begin{proof}[Proof of Proposition~\ref{prop:finite_bailout}]
Consider a cash infusion $\iota(x)$. By Proposition~\ref{prop:bailout}, $f_k(x) + \iota(x)$ is weakly increasing with at most one discontinuity in each block $k$. Hence, we can choose smooth and strictly increasing functions $g_k(x)$ such that $g_k(x) \geq f_k(x) + \iota(x)$ for all $x$  and
$$\sum_k\int_{x \in \BT_k} g_k(x) dx \leq \sum_k\int_{x \in \BT_k} (f_k(x) + \iota(x)) dx+ \epsilon.$$
This inequality implies that the functions $g_k(x)-f_k(x)$ define a cash infusion with budget at most $K+\epsilon.$ Let $\widetilde{x}_k^*$ be the cutoffs at the maximal equilibrium after cash infusion $\iota(x)$ and let $\widetilde{y}_k^*$ be the cutoffs at the maximal equilibrium after cash infusion $g_k(x)-f_k(x)$. It follows from Lemma~\ref{lem:eq_construction} that increasing firm endowments cannot diminish the set of firms that are solvent at the maximal equilibrium, so we have $\widetilde{y}_k^* \leq \widetilde{x}_k^*$ for all $k$.

Define a cash infusion $\mathbf{\iota}^{(n)}$ with $n$ agents by setting $\iota^{(n)}_i = g_k(i) - f_k(i)$ for each $i \in \BT_k$. As in Theorem \ref{thm:maximaleq}, let $\overline{i}_k(n)$ be the maximal index of a firm of type $k$ that is insolvent under the maximal equilibrium after cash infusion $\mathbf{\iota}^{(n)}$. By Theorem \ref{thm:maximaleq}, we have $\overline{i}_k(n) \rightarrow \widetilde{y}_k^* $ for all $k$. So, the limit inferior of the fraction of firms solvent under the maximal equilibrium after cash infusion $\mathbf{\iota}^{(n)}$ is at least
$$\sum_k (t_k -\widetilde{y}_k^*)  \geq \sum_k (t_k -\widetilde{x}_k^*)
 = \lambda(\{x \in (0,1]:\widetilde{v}_x \geq v^* \}),
$$
which proves the proposition.
\end{proof}

\end{document}